\documentclass[12pt,twoside]{article}
\pdfoutput=0
\usepackage{amssymb}
\usepackage{amsmath}
\usepackage[makeroom]{cancel}
\usepackage{latexsym}
\usepackage{longtable}
\usepackage{graphicx,bm,psfrag}
\usepackage{graphicx}
\graphicspath{{images/}}

\newtheorem{theorem}{Theorem}[section]
\newenvironment{proof}{{\bf Proof}:}{\hspace*{\fill}$\Box$}

\renewcommand{\thetheorem}{\arabic{section}.\arabic{theorem}}

\newtheorem{lemma}[theorem]{Lemma}
\newtheorem{pro}[theorem]{Proposition} 
\newtheorem{cor}[theorem]{Corollary} 
\newtheorem{con}[theorem]{Conjecture} 

\begin{document}
\begin{titlepage}
\begin{center}
\vspace*{2cm} {\Large {\bf  Point-interacting Brownian motions in the KPZ universality class \bigskip\bigskip\\}}
{\large Tomohiro Sasamoto$^{1,2}$ and Herbert Spohn$^1$}\bigskip\bigskip\\
 $^1$Zentrum Mathematik and Physik Department, TU M\"unchen,\\
 Boltzmannstr. 3, D-85747 Garching, Germany\medskip\\
$^2$Department of Physics, Tokyo Institute of Technology\\ 
2-12-1 Ookayama, Meguro-ku, Tokyo, 152-8550, Japan\medskip
\\email:~{\tt sasamoto@phys.titech.ac.jp, spohn@ma.tum.de
}\end{center}
\vspace{5cm} \textbf{Abstract.} We discuss chains of interacting Brownian motions. Their time reversal invariance is broken because of asymmetry in the interaction strength between left and right neighbor. In the limit
of  a very steep and short range potential one arrives at Brownian motions with oblique reflections. For this model we prove a Bethe ansatz formula for the transition probability and self-duality. In case of half-Poisson initial data, duality is used to arrive at a Fredholm determinant for
the generating  function of the number of particles to the left of some reference point at any time $t > 0$.
A formal asymptotics for this determinant establishes the link to the Kardar-Parisi-Zhang universality class.
\end{titlepage}

\section{Nonreversible interacting diffusions}\label{sec1}
\setcounter{equation}{0}
Roughly fifteen years ago, K. Johansson established that the totally asymmetric simple exclusion process (TASEP) is in the Kardar-Parisi-Zhang (KPZ) universality class. More precisely, for step initial conditions he studied  $J_{0,1}(t)$, the particle current between sites 0 and 1,
integrated over the time span $[0,t]$, and proved that 
\begin{equation}\label{1.1}
J_{0,1}(t) = c_{\mathrm{d}}t + c_{\mathrm{f}}t^{1/3}\xi_\mathrm{{GUE}}
\end{equation}
in distribution for large $t$. The random amplitude $\xi_\mathrm{{GUE}}$ is GUE Tracy-Widom distributed. $c_{\mathrm{d}},
c_{\mathrm{f}}$ are explicitly known  constants, but  
to keep the notation light we do not display them here. These are model dependent, non-universal coefficients, which will reappear again
and may take different numerical values. (`d' stands for deterministic and `f' for fluctuations). The scaling exponent $1/3$ was predicted before by Kardar, Parisi, and 
Zhang \cite{KPZ}, see also \cite{HH85,vBKS85}. The most striking feature is the random amplitude, telling us that 
(\ref{1.1}) is not a central limit theorem. Many related results have been established since, for surveys see \cite{Co12,BG12,BP13,QR13}.
Most of them are for specific interacting stochastic particle systems in one dimension, which are discrete and have a dynamics governed by a Markov jump process. In this contribution we will explore interacting one-dimensional diffusion processes
in the KPZ universality class.

As a start we define a family of model systems, explain in more detail the conjectures related to the KPZ universality class, and recall the two major results available so far. The main part of our contribution concerns a 
singular limit, in which the Brownian motions  interact only when they are at the same location.  

To motivate our model system we start from the potential of a coupled chain,
\begin{equation}\label{1.0} 
V_\mathrm{tot}(x) = \sum_{j=1}^{n-1} V(x_{j+1} - x_{j})
\end{equation}
with $x = (x_1,...,x_n)$, $x_j \in \mathbb{R}$, and a twice differentiable nearest neighbor potential, $V$.  To construct a reversible diffusion process with invariant measure
\begin{equation}\label{1.7}
\mathrm{e}^{-V_\mathrm{tot}(x)} \prod _{j=1}^n \mathrm{d}x_j\,,
\end{equation}
the drift is taken to be the gradient of $V_\mathrm{tot}$, while the noise is white and  independent for each coordinate. Then
\begin{equation}\label{1.2}
\mathrm{d}x_j(t) = \big(\tfrac{1}{2}V'( x_{j+1}(t)- x_{j}(t)) - \tfrac{1}{2}V'(x_{j}(t) - x_{j-1}(t))\big)\mathrm{d}t + \mathrm{d}B_j(t)\,,
\end{equation}
$j = 1,....,n$, with the convention that $V'(x_{1}(t) - x_{0}(t))= 0 = V'(x_{n+1}(t) - x_{n}(t))$.  Here $x_j(t) \in \mathbb{R}$ and $\{B_j(t), j= 1,....,n\}$
is a collection of independent standard Brownian motions. Note that the measure in (\ref{1.7}) has infinite mass. 

The dynamics defined by (\ref{1.2}) is invariant under the shift $x_j \leadsto x_j +a$, which will be the origin for slow decay in time. Breaking this shift invariance, for example by adding an external, confining on-site potential $V_\mathrm{ex}$ as $-V_\mathrm{ex}'(x_j(t))\mathrm{d}t$ in (\ref{1.2}), would change the picture completely. 
Just to give one example, one could choose $V$ and $V_\mathrm{ex}$ to be quadratic. 
Then the dynamics governed by Eq. (\ref{1.2}) is an Ornstein-Uhlenbeck process, which has a unique invariant measure, a spectral gap independent of system size, and exponential 
space-time mixing. Setting $V_\mathrm{ex}= 0$, slow decay is regained. Because of shift invariance, we regard $x_j(t)$ as the height at lattice site $j$ at time $t$. In applications $x_j(t)$ could describe a one-dimensional interface which separates two bulk phases of a thin film of a binary liquid  mixture. $V$ then models  the surface free energy (surface tension) of this interface.

If in (\ref{1.2}) one introduces the stretch $r_j = x_j - x_{j-1}$ and adopts periodic boundary conditions, then 
\begin{equation}\label{1.4}
\mathrm{d}r_j(t) = \tfrac{1}{2}\Delta V'(r_j(t))\mathrm{d}t + \nabla \mathrm{d}B_j(t)\,,\quad j=1,...,n\,,
\end{equation}
where $\Delta$ denotes the lattice Laplacian and $\nabla$ the finite difference operator, both understood with periodic boundary conditions. 
Clearly, $r_j(t)$ is locally conserved and the sum $\sum_{j=1}^n r_j(t)$ is conserved. As a consequence the $r(t)$ process has a one-parameter family of invariant probability measures,  indexed by $\ell$, which is obtained by conditioning the measure
\begin{equation}\label{1.7a}
\prod _{j=1}^n \mathrm{e}^{-V(r_{j})} \mathrm{d}r_j\,,
\end{equation}
on the hyperplane $\{r\,|\,\sum_{j=1}^n r_j= n\ell\}$. 
In the infinite volume limit, the $\{r_j\}$ are  i.i.d. with the single site distribution
\begin{equation}\label{1.5}
Z^{-1}\mathrm{e}^{-V(r_j) - Pr_j} \mathrm{d} r_j\,,\quad Z = \int \mathrm{e}^{-V(u) - Pu} \mathrm{d}u\,,\quad
\mathbb{E}_P(r_j) = \ell\,,
\end{equation}
where $\mathbb{E}_P(\cdot)$ denotes expectation with respect to the product measure. The parameter $P$ controls the average value of $r_j$. To have  $Z < \infty$ for a nonempty interval of values of $P$, we require the potential $V$ to be bounded from below and to have at least a one-sided bound as $V(u) \geq c_1 + c_2|u|$, either for $u > 0$ or for $u < 0$,
with $c_2 > 0$. Note that 
\begin{equation}\label{1.6}
-\mathbb{E}_P(V'(r_j)) = P\,,
\end{equation}
which means that $P$ is the equilibrium pressure in the chain. The diffusive limit of (\ref{1.4}) has been studied in a famous work
by Guo, Papanicolaou, and Varadhan \cite{GPV88}, who prove that on a large space-time scale the random field $\{r_j(t), j = 1,...,n\}$ is well approximated by a deterministic nonlinear diffusion
equation.
The fluctuations relative to the deterministic space-time profile  are Gaussian as proved by Chang and Yau \cite{CY92}. 

KPZ universality enters the play, when the dynamics (\ref{1.2}) is modified to become nonreversible. 
In the physical picture of an interface, the breaking of time reversal invariance results from an imbalance between the two bulk phases which induces  a systematic motion. On a more abstract level there are many options.
One possibility is to start from a Gaussian process by setting $V(u) = u^2$
and adding nonlinearities such that shift invariance is maintained and  the stationary Gaussian measure  of the linear
equations remains stationary,
see \cite{SS09} for a worked out example. Here we take a different route by splitting the two drift  terms,
$\tfrac{1}{2}V'(x_{j+1}(t) - x_{j}(t))$ and $ - \tfrac{1}{2} V'( x_{j}(t) - x_{j-1}(t))$, not symmetrically but
asymmetrically with fraction $p$ to the right and fraction $q$ to the left, $p+q = 1$, $0 \leq p \leq 1$.
Then (\ref{1.2}) turns into
\begin{equation}\label{1.8}
\mathrm{d}x_j(t) = \big(pV'( x_{j+1}(t)- x_{j}(t)) - qV'(x_{j}(t) - x_{j-1}(t))\big)\mathrm{d}t + \mathrm{d}B_j(t)\,.
\end{equation}
The totally asymmetric limits correspond to $p = 0,1$. One easily checks that 
for all $p$ the measure (\ref{1.7}) is still invariant which, of course, is a good reason to break time reversal invariance in this particular way.  This property  is in analogy to the ASEP, where the Bernoulli measures are invariant independently of the choice of the right hopping rate $p$. 

If, as before, one switches to the stretches $r_j$, then 
\begin{equation}\label{1.8a}
\mathrm{d}r_j(t) = \tfrac{1}{2}\Delta_\mathrm{p} V'(r_j(t))\mathrm{d}t + \nabla \mathrm{d}B_j(t)\,,\quad j=1,...,n\,,
\end{equation}
with periodic boundary conditions and $\tfrac{1}{2}\Delta_\mathrm{p}f(j) = pf(j+1) + q f(j-1) - f(j)$.
Because of the asymmetry, the macroscopic scale is hyperbolic rather than diffusive.
We denote by $\ell(u,t)$ the macroscopic field for the local stretch $r_j(t)$, where $u$ is the continuum limit 
of the labeling by lattice sites $j$. Then, using the entropy method of Yau \cite{Y91}, it can be proved that the deterministic limit satisfies the hyperbolic conservation law
\begin{equation}\label{1.10}
\partial_t \ell + (p-q)\partial_u P(\ell) =0
\end{equation}
 with $P(\ell)$ the function inverse to $\mathbb{E}_P(r_0) = \ell$.
 Since $P'(\ell) < 0$,  the inverse is well defined. The limit result leading to (\ref{1.10})
holds for initial profiles which are slowly varying on the scale of the lattice and up to the first time when a shock 
is formed. 

At this point we can explain the striking difference between reversible and nonreversible systems. Let us impose the periodic initial configuration $x_j = \bar{\ell} j$, $j \in \mathbb{Z}$ and, assuming that the dynamics for the infinite system is well defined,
let us  focus on $x_0(t)$, the particle starting at the origin. For the symmetric model one expects
\begin{equation}\label{1.11}
x_0(t) = c_{\mathrm{f}} t^{1/4} \xi_\mathrm{G}
\end{equation}
as $t \to \infty$ with $\xi_\mathrm{G}$ a standard mean zero Gaussian random variable. We are not aware of a completely
written out proof, but the key elements can be found in \cite{CY92}. Harris \cite{Har65} considers independent Brownian motions,
such that the labeling is maintained according to their order. For the dynamics defined by (\ref{1.2}) this corresponds to the limit of a strongly repulsive
potential $V$ with its support shrinking to zero. In \cite{Har65} it is proved that $x_0(t)$ is well-defined and that the scaled process $\epsilon^{1/4}x_0(\epsilon^{-1}t)$ has a limit as $\epsilon \to 0$ which is a Gaussian process with an explicitly computed covariance.

In contrast, for the nonreversible system it is conjectured that
\begin{equation}\label{1.12}
x_0(t) = (p-q)P(\bar{\ell})t +  c_{\mathrm{f}}t^{1/3} \xi_\mathrm{GOE}
\end{equation}
in distribution as $t \to \infty$. The anticipated numerical value of $c_{\mathrm{f}}$ is explained in Appendix \ref{app.c}. 
Note that, in general,  there could be specific values of $\bar{\ell}$, for which $c_{\mathrm{f}} =0$. In particular, for the Gaussian process with
$V(u) = \tfrac{1}{2} u^2$, one obtains $P(\ell) = \ell$ and $c_{\mathrm{f}} =0$ for all $\ell$. The random amplitude
$ \xi_\mathrm{GOE}$ has the distribution function
\begin{equation}\label{1.13}
\mathbb{P}(\xi_\mathrm{GOE} \leq s) = \det(1 - P_sB_0P_s)\,.
\end{equation}
Here the determinant is over $L^2(\mathbb{R})$, $P_s$ projects onto the half-line $[s,\infty)$, and  $B_0$ is a Hermitean operator
with integral kernel $B_0(u,u') = \mathrm{Ai}(u+u')$, $\mathrm{Ai}$ being the standard Airy function. 
As proved by Tracy and Widom \cite{TW94}, the expression (\ref{1.13}) is also the 
distribution function of the largest eigenvalue of the Gaussian Orthogonal Ensemble (GOE) of  real symmetric $N\times N$ random matrices in the limit $N \to \infty$,  see \cite{S04,FS05} for the particular representation (\ref{1.13}).

As in the case of a reversible model, one can regard $x_0(t)$ as a stochastic process in $t$. No definite conjectures on its scaling limit are available. We refer to \cite{CQ11} for a discussion.

A proof of (\ref{1.12}) seems to be difficult with current techniques, except for the Harris limiting case with $q = 1$. Then the process
$\{x_j(t), j \in \mathbb{Z}\}$ is constructed in the following way:  for all $j$, $x_j(0) = j$ and $x_j(t)$ performs a Brownian motion being reflected at the Brownian particle $x_{j-1}(t)$. Because of collisions, $x_0(t)$ is pushed to the right and,
as proved in \cite{FSW13}, it holds that
\begin{equation}\label{1.14}
\lim_{t \to \infty}(2t)^{-1/3}(x_0(t) - t) = \xi_\mathrm{GOE} 
\end{equation}
in distribution.

There is a second example which can be analysed in considerable detail and again confirms anomalous fluctuations. 
As before the dynamics is totally asymmetric, $q = 1$, but the potential is smooth and given by $V(u) = \mathrm{e}^{-u}$. Then
\begin{equation}\label{1.15}
x_0(t) = x_0(0) + B_0(t)\,,\quad      \mathrm{d}x_j(t) =  \exp\!\big(-x_{j}(t) + x_{j-1}(t)\big)\mathrm{d}t +
     \mathrm{d} B_j(t)\,, \,j = 1,2,...\,. 
\end{equation}
The initial conditions are $x_0(0) = 0$ and, formally, $x_j(0)= - \infty$ for $j \geq 1$. As proved in \cite{MO07},
there is a law of large numbers which states that
\begin{equation}\label{1.16}
\lim_{t \to \infty} t^{-1}x_{\lfloor ut \rfloor}(t) = \phi(u)\,\,\mathrm{a.s.}
\end{equation}
 for $u > 0$ with $\lfloor \cdot \rfloor$ denoting integer part. The limit function $\phi$ can be guessed by realizing that on the macroscopic scale the slope satisfies Eq. (\ref{1.10}). 
 First note that $\ell = - \psi(P)$ with $\psi = \Gamma'/\Gamma$, the Digamma function. Hence
 \begin{equation}\label{1.16a}
 \phi(u) = \inf_{s \geq 0} \big(s - u \psi(s)\big)\,,
 \end{equation}
 see \cite{Sp12} for details. $\phi(0) = 0$, $\phi''<0$, and $\phi$ has a single strictly positive maximum
before dropping to $-\infty$ as $u\to\infty$. Thus $t\phi(u/t)$ reproduces the required singular initial conditions 
as $t\to 0$.

 Even more remarkable, one has a limit result \cite{BC13,BCF13} for the fluctuations, 
 \begin{equation}\label{1.17}
\lim_{t \to \infty} t^{-1/3}\big(x_{\lfloor ut \rfloor}(t)  - t \phi(u)\big) = \kappa(u)^{1/3} \xi_{\mathrm{GUE}} \,.
\end{equation}
The non-universal coefficient $\kappa(u)$ will be discussed in Appendix \ref{app.c}. Note that the proper rule is to subtract the asymptotic mean value and not the more obvious mean at time $t$. In fact $\mathbb{E}( \xi_{\mathrm{GUE}} ) = -1.77.$

In our contribution we will study interacting diffusions with \textit{partial} asymmetry and random initial data. As in the previous example, the index $j \in \mathbb{Z}_+$. But we  
have to resort to point interactions. The precise definition of the dynamics will be given in the following section.
As initial conditions we assume that $\{x_0, x_{j+1} - x_j, j \geq 0\}$ are independent exponentially distributed random variables with mean $1$. Hence at $t=0$ the macroscopic profile  is $\phi(u) = u$, $u\geq 0$. For point interactions,
one has $V=0$ in Eq. (\ref{1.5}) and thus $P(\ell) = \ell^{-1}$. The integrated version of  Eq. (\ref{1.10}) reads 
\begin{equation}\label{1.18a}
\partial_t \phi + (p-q)(\partial_u \phi)^{-1} = 0\,,  
\end{equation}
which for our initial conditions has the self-similar solution
\begin{equation}\label{1.18}
\phi(u,t/\gamma) = 2 \sqrt{ut}\,\,\,\mathrm{for}\,\,\,0 \leq u \leq t\,, \quad \phi(u,t) = u+t \, \,\,\mathrm{for}\,\,t \leq u
\end{equation}
with $p < 1/2$, $\gamma = q - p$. Anomalous fluctuations are expected to be seen in the window $0 < u < \gamma t$ not
too close to the boundary points.

The three examples discussed above require distinct techniques in their analysis. The first example  uses that, upon judiciously choosing
dummy variables, there is an embedding signed determinantal process.
In the second example one derives
a Fredholm determinant for the generating function $\mathbb{E}\big(\exp[-\zeta\mathrm{e}^{x_j(t)}]\big)$ with $\zeta \in \mathbb{C}$, $\Re{\zeta} > 0$. In contrast our analysis is based on self-duality
of the particle system. $x_j(t)$ is replaced by  $N(u,t)$, which is the number of particles to the left of $u$ at time $t$, 
\textit{i.e.} the largest $j$ such that $x_j(t) \leq u$. $\mathrm{e}$ is replaced by $\tau = p/q <1$ and $\exp$ 
by the  $\tau$-deformed  exponential $e_\tau$.
Following the strategy in \cite{BCS14}, we arrive at a Fredholm determinant for the expectation  $\mathbb{E}\big(e_\tau(
\zeta\tau^{N(u,t)})\big)$. This is our main result. To establish the
connection to KPZ universality, we add a heuristic discussion of a saddle point analysis for this Fredholm determinant.
To prove duality we need some information on the transition probability, which will be provided in a form following from the Bethe ansatz. Such a formula could be of use also in other applications.\\\\ 
\textbf{Acknowledgements}. We thank for the warm  hospitality at the Institute for Advanced Study at Princeton, where the major part of our work was completed. We thank Tadahisa Funaki for advice concerning \ref{app.a}. 
HS thanks Jeremy Quastel for most constructive discussions and Thomas Weiss for helping with the figures. TS is grateful for the support from KAKENHI 22740054 and Sumitomo Foundation.

\section{Brownian motions with point interactions,\\ self-duality}\label{sec2}
\setcounter{equation}{0}

We consider $n$ interacting Brownian particles governed by the asymmetric dynamics of Eq. (\ref{1.8}). Point interactions are realized through a sequence of potentials, $V_\epsilon$, which are repulsive, diverge sufficiently rapidly as $|u|\to 0$, and whose range shrinks to zero as $\epsilon \to 0$. More precisely, we start from a reference potential $V \in C^2(\mathbb{R} \setminus \{0\}, \mathbb{R}_+)$  with the properties $V(u) = V(-u)$,
$\mathrm{supp} \,V = [-1,1]$, $V'(u) \leq 0$ for $u > 0$, and, for some $\delta > 0$, $\lim_{u \to 0} |u|^{\delta}V(u) > 0$. 
The scaled potential is defined by $V_\epsilon(u) = V(u/\epsilon)$ and the corresponding diffusion process is denoted by $y^\epsilon(t)$. Since the potential is entrance - no exit \cite{MK}, the positions can be ordered as $y^\epsilon_1(t)\leq ...
\leq y^\epsilon_m(t)$. Hence $y^\epsilon(t) \in \mathbb{W}_m^+$, the Weyl chamber in $\mathbb{R}^m$ such that the left-right order is according to increasing index. Since the
particle order is preserved, we deviate slightly from the viewpoint of the introduction and regard the positions of particles  as a point configuration in $\mathbb{R}$. As will be proved in Appendix \ref{app.a}, there exists a limit process, $y(t) \in \mathbb{W}_m^+$,
such that $\lim_{\epsilon \to 0} y^\epsilon(t) = y(t)$. Presumably the limit holds a.s. in the sup norm, but for our purposes it suffices to 
prove that $\lim_{\epsilon \to 0}\mathbb{E}\big((y^\epsilon(t)- y(t))^2\big) = 0$. The limit process $y(t)$ is Brownian motion with point interaction, also known as  Brownian motion  with oblique reflection.

$y(t)$ is a semi-martingale satisfying
\begin{equation}\label{2.1}
y_j(t) = y_j + B_j(t) - p \Lambda^{(j,j+1)}(t) +q \Lambda^{(j-1,j)}(t)\,, 
\end{equation}
$t\geq 0$, $j = 1,...,m$. Here $p+q = 1$, $0 \leq p \leq 1$, and by definition $\Lambda^{(0,1)}(t) = 0 =  \Lambda^{(m,m+1)}(t)$.
\begin{equation}\label{2.2}
\Lambda^{(j,j+1)}(\cdot) = L^{y_{j+1} -y_j}(\cdot,0) 
\end{equation}
is the right-sided local time accumulated at the origin by the nonnegative martingale $y_{j+1}(\cdot) -y_j(\cdot)$.
So $y_j(t)$ is pushed to the left with fraction $p$ of the local time whenever $y_j(t) = y_{j+1}(t)$ and it is pushed
to the right with fraction $q$ of the local time whenever $y_j(t) = y_{j-1}(t)$, which implies that the drift always pushes  towards the interior of $\mathbb{W}_m^+$. If $q=1$,  $y_{j+1}(t)$ is reflected at $y_j(t)$.
In particular, $y_1(t)$ is Brownian motion. If $q = 1/2$, the dynamics corresponds to independent Brownian motions with ordering of labels maintained. In \cite{KPS12} it is proved that (\ref{2.1}) has a unique strong solution. Furthermore, triple collisions,
\textit{i.e.} the sets $\{y_j(t) = y_{j+1}(t) = y_{j+2}(t) \,\,\mathrm{for\,\, some}\,\, t\}$, have probability 0.
 
 Let $f:  \mathbb{W}_m^+ \to \mathbb{R}$ be a $C^2$-function and define
 \begin{equation}\label{2.3}
f(y,t) = \mathbb{E}_y\big(f(y(t)\big)
 \end{equation} 
with $\mathbb{E}_y$ denoting expectation of the $y(t)$ process of (\ref{2.1}) starting at $y \in \mathbb{W}_m^+$. As proved in Section \ref{sec6}, it holds
\begin{equation}\label{2.4}
\partial_t f = \tfrac{1}{2} \Delta_y f
\end{equation} 
for $ y \in (\mathbb{W}_m^+)^\circ$ and
\begin{equation}\label{2.5}
(p\partial_j - q\partial_{j+1})f \big |_{y_j = y_{j+1}} = 0\,,
\end{equation} 
the directional derivative being taken from the interior of $\mathbb{W}_m^+$.
$q = 1/2$ corresponds to normal reflection at $\partial \mathbb{W}_m^+$. With this boundary condition
$\Delta_y$ is a self-adjoint operator. $q\neq 1/2$ is also referred to as oblique reflection at $\partial \mathbb{W}_m^+$
\cite{VW84,HW87}. 

In addition to the $y$-particles we introduce $n$ dual particles denoted by $(x_1(t),...,x_n(t))$ $ = x(t)$. They are ordered as
$x_n \leq ... \leq x_1$, hence $x \in \mathbb{W}_n^-$, the Weyl chamber in $\mathbb{R}^n$ such that the left-right order is according to decreasing index. For the dual particles the role of $q$ and $p$ is interchanged. Thus their dynamics is still governed by (\ref{2.1}) with
$\Lambda^{(j,j+1)}(\cdot) = L^{x_{j} -x_{j+1}}(\cdot,0)$. Also the boundary condition (\ref{2.5}) remains valid, 
the directional derivative being taken from the interior of $\mathbb{W}_n^-$.

The main goal of this section is to establish that the $x(t)$ process is dual to the $y(t)$ process. The duality function is defined by
\begin{equation}\label{2.6}
H(x,y) = \prod_{j=1}^n \prod_{i=1}^m \tau^{\theta(x_j - y_i)}\,,
\end{equation} 
where $\tau = p/q$ and throughout we  restrict to the case $0< \tau < 1$. $\theta(u) = 0$ for $u \leq 0$ and $\theta(u) = 1$
for $u >0$. Such type of duality is known also for other stochastic particle systems \cite{JK14}, in particular for the ASEP
\cite{BCS14}. 
\begin{theorem}\label{th1}
Pointwise on $\mathbb{W}_n^- \times \mathbb{W}_m^+$ it holds
\begin{equation}\label{2.8}
\mathbb{E}_x\big(H(x(t),y)\big)
 = \mathbb{E}_y\big(H(x,y(t))\big)  \,.             
\end{equation} 
\end{theorem}
\begin{proof}
We first compute the distributional derivative of $H$.  Setting $\partial_{x_\alpha} = \partial/\partial x_\alpha$ for $\alpha = 1,...,n$ one obtains
\begin{eqnarray}\label{2.9}
&&\hspace{-58pt}\partial_{ x_\alpha}H(x,y)  = -(1-\tau) \sum_{\beta = 1}^m \delta(x_\alpha - y_\beta) \prod_{i'= 1}^m 
\tau^{\theta(x_\alpha - y_{i'})} \mathop{\prod^n_{j = 1}}_{j\neq\alpha} \prod_{i=1}^m \tau^{\theta(x_j - y_{i}) }\nonumber\\
&&\hspace{0pt}= -(1 - \tau) \sum_{\beta = 1}^m \delta(x_\alpha - y_\beta) \mathop{\prod_{j' =1}^n}_{j'\neq\alpha}
\tau^{\theta(x_{j'} - x_{\alpha})}\mathop{\prod_{i'=1}^m}_{i'\neq\beta} \tau^{\theta(x_{i'} - y_{\beta})}
\mathop{\prod_{j = 1}^n}_{ j\neq\alpha} \mathop{\prod_{i=1}^m}_{i^{}\neq\beta} \tau^{\theta(x_j - y_{i}) }\nonumber\\
&&\hspace{0pt}= -(1-\tau)  \sum_{\beta = 1}^m \delta(x_\alpha - y_\beta)\tau^{\beta -1}\tau^{\alpha -1}
\mathop{\prod_{j = 1}^n}_{ j\neq\alpha} \mathop{\prod_{i=1}^m}_{ i\neq\beta} \tau^{\theta(x_j - y_{i}) }
 \end{eqnarray}
and correspondingly
\begin{equation}\label{2.10}
\partial_{y_{\beta}}H(x,y)  = (1-\tau)  \sum_{\alpha = 1}^n \delta(x_\alpha - y_\beta)\tau^{\beta -1}\tau^{\alpha -1}
\mathop{\prod_{j = 1}^n}_{ j\neq\alpha} \mathop{\prod_{i=1}^m}_{ i\neq\beta} \tau^{\theta(x_j - y_{i}) }\,.
 \end{equation}
Let us set $\mathcal{D}(L_x) = C^2_{0,\mathrm{bc}}(\mathbb{W}^-_n, \mathbb{R})$, the set of all twice continuously differentiable functions vanishing rapidly at infinity and
with boundary conditions
\begin{equation}\label{2.11}
(p\partial_j - q\partial_{j+1})f \big |_{x_j = x_{j+1}} = 0\,.
\end{equation} 
As will be discussed in Section \ref{sec6}, the generator $L_x$ of the diffusion process $x(t)$ is given by  $L_x = \tfrac{1}{2} \Delta_x$ on the domain $\mathcal{D}(L_x)$ and correspondingly for $L_y$. The integral kernel of $\mathrm{e}^{L_x t}$, denoted by $P_{x}^-(\mathrm{d}x',t)$,
is the transition probability for $x(t)$. It has a density, $P^-_x(\mathrm{d}x',t)  = P^-_x(x',t)\mathrm{d}x' $.
$ P_x^-(x',t)$ is $C^{\infty}$ in both $x,x'$ when restricted to the set $\big(\mathbb{W}^-_n\setminus \{ x\,|\,x_j = x_{j+1} = x_{j+2}, j = 1,...,n-2\}\big)^{\times 2}$. 
\begin{lemma}\label{le1}
Let $f \in  C^2_0(\mathbb{W}^+_m, \mathbb{R})$ and define
\begin{equation}\label{2.12}
F(x) = \int_{\mathbb{W}^+_m} H(x,y)f(y)\mathrm{d}y\,.
\end{equation} 
Then $F \in  \mathcal{D}(L_x)$.
\end{lemma}
\begin{proof} Since $H$ is a product of convolutions, $F \in C^2_0(\mathbb{W}^-_n, \mathbb{R})$. We use (\ref{2.9}) for $\alpha = j,j+1$. Then 
\begin{equation}\label{2.13}
(\tau\partial_j - \partial_{j+1})F\big |_{x_j = x_{j+1}} = 0\,.
\end{equation}
\end{proof} 

By the fundamental theorem of calculus, for $0< \epsilon < t - \epsilon$,
\begin{eqnarray}\label{2.14}
&&\hspace{-80pt}\big(\mathrm{e}^{L_x(t-\epsilon)}\otimes
\mathrm{e}^{L_y\epsilon}H\big)(x,y) - \big(\mathrm{e}^{L_x\epsilon}\otimes
\mathrm{e}^{L_y(t-\epsilon)}H\big)(x,y)\nonumber\\
&&\hspace{30pt}
= \int_\epsilon^{t-\epsilon} \mathrm{d}s \frac{\mathrm{d}}{\mathrm{d}s}\big(\mathrm{e}^{L_x s}\otimes
\mathrm{e}^{L_y(t-s)}H\big)(x,y)\,.
\end{eqnarray} 
By Lemma \ref{le1} and for $\epsilon \leq s \leq t -\epsilon$ the function
\begin{equation}\label{2.15}
x \mapsto \int_{\mathbb{W}^+_m} \mathrm{d}y' P^+_y(y',s)H(x,y') \in \mathcal{D}(L_x)
\end{equation} 
and correspondingly for $y$. Hence one can differentiate in (\ref{2.14}) and obtains
\begin{eqnarray}\label{2.16}
&&\hspace{-40pt}\big(\mathrm{e}^{L_x(t-\epsilon)}\otimes
\mathrm{e}^{L_y\epsilon}H\big)(x,y) - \big(\mathrm{e}^{L_x\epsilon}\otimes
\mathrm{e}^{L_y(t-\epsilon)}H\big)(x,y)\nonumber\\
&&\hspace{-20pt}
= \int_\epsilon^{t-\epsilon} \!\!\!\mathrm{d}s 
\int_{\mathbb{W}^-_n}  \mathrm{d}x' \int_{\mathbb{W}^+_m}\mathrm{d}y' P^-_x(x',s)  P^+_y(y',t-s) 
\big(L_xH(x',y') - L_yH(x',y')\big)\,.
\end{eqnarray} 
Since the transition probabilities are smooth, $L_xH$ and $L_yH$ can be obtained as distributional derivatives. Hence
\begin{equation}\label{2.17}
\Delta_x H(x,y) = -(1-\tau) \sum_{\alpha = 1}^n \sum_{\beta = 1}^m \delta'(x_\alpha - y_\beta)\tau^{\beta -1}\tau^{\alpha -1}
\mathop{\prod_{j = 1}^n}_{ j\neq\alpha} \mathop{\prod_{i=1}^m}_{ i\neq\beta} \tau^{\theta(x_j - y_{i}) }
= \Delta_y H(x,y)
\end{equation}
and 
\begin{equation}\label{2.18}
\big(\mathrm{e}^{L_x(t-\epsilon)}\otimes
\mathrm{e}^{L_y\epsilon}H\big)(x,y) = \big(\mathrm{e}^{L_x\epsilon}\otimes
\mathrm{e}^{L_y(t-\epsilon)}H\big)(x,y)\,.
\end{equation} 

We integrate Eq. (\ref{2.18}) against the smooth function $f_1(x)f_2(y)$. By continuity we can take the limit $\epsilon \to
0$. The integrand of the resulting identity  is continuous in $x,y$ and the identity (\ref{2.8}) holds pointwise. \bigskip
\end{proof}\\
\textit{Remark}. An alternative proof, based on ASEP duality, is discussed in Appendix \ref{app.b}.


\section{Half-line Poisson as initial conditions, contour integrations}\label{sec3}
\setcounter{equation}{0}
We assume that initially the particles are Poisson distributed with density profile $\rho(u) = \theta(u)$. By space-time scaling, 
the density $1$ on the half-line can be changed to any other value. Let us denote by $N(u;y)$ the number of particles in the configuration $y$ located in $(-\infty,u]$ and set $N(u,t) = N(u;y(t))$ as a random variable. We average the duality function over the Poisson distribution,
\begin{eqnarray}\label{1}
&&\hspace{-30pt}\mathbb{E}_{\mathrm{poi}}\big(H(x,\cdot)\big) = \mathbb{E}_{\mathrm{poi}}\big(\prod_{j=1}^n 
\prod_{i=1}^\infty 
\tau^{\theta(x_j - y_i)}\big) = \mathbb{E}_{\mathrm{poi}}\big(\prod_{j=1}^n \tau^{N(x_j;y)}\big)\nonumber\\
&&\hspace{41pt}= \exp\Big[ \int_0^\infty \!\!\mathrm{d}u\big( \prod_{j=1}^n  \tau^{\theta(x_j - u)} -1\big)\Big] = F_n(x)\,.
\end{eqnarray} 
Next the duality relation (\ref{2.8}) is averaged  over the Poisson distribution with the result \begin{equation}\label{2}
F_n(x,t) = \mathbb{E}_x\big(F_n(x(t))\big) =\mathbb{E} \big(\prod_{j=1}^n \tau^{N(x_j,t)}\big)\,.
\end{equation} 
Here $\mathbb{E}$ refers to the particle process with half-line Poisson as initial measure. We regard the right-hand generating function as defined through the left-hand side.
It can be obtained by first considering a Poisson measure with density $\rho(u) = 1$ in the interval $[0,L]$
and zero outside. Then the Poisson average in (\ref{2}) is well defined. Taking the limit $L\to \infty$ yields the left-hand side of (\ref{2}).

More ambitiously, one should first define the $y(t)$ process for an infinite number of particles, in such a way that it supports the Poisson measure. Thereby the random variable $N(u,t)$ would be well-defined. In particular, for our initial measure,
$\mathbb{P}\big(\{N(u,t) = \infty\}\big) = 0$.  

The next step is to arrive at a contour integration formula for $F_n$.
\begin{theorem}\label{th3}
With $F_n$ from (\ref{2}) one has
\begin{equation}\label{3}
F_n(x,t) =  \tau^{n(n-1)/2} \int_{\mathcal{C}} \mathrm{d}z_1 ... \mathrm{d}z_n
\prod_{j=1}^n \frac{1}{z_j}\cdot \frac{\tau -1}{z_j + (1-\tau)}\, \mathrm{e}^{x_jz_j + \frac{1}{2}tz_j^2}
\prod_{1\leq A < B \leq n} \frac{z_B - z_A}{z_B - \tau z_A}\,.
\end{equation} 
The contours are $\mathcal{C}_j = \{a_j + \mathrm{i}\varphi, \varphi \in \mathbb{R}\}$ and  nested as
$-(1-\tau) < a_1< ...< a_n < 0$ such that $\tau a_j < a_{j+1}$.
\end{theorem}
\textit{Remark}. It is understood throughout that the contour integration includes the prefactor $1/2\pi\mathrm{i}$.\\\\
\begin{proof}
Let us denote the right hand side of (3.3) by $\tilde{F}_n(x,t)$. We have to show that $F_n(x,t) = \tilde{F}_n(x,t)$.\medskip\\
\textit{(i) evolution equation.} By inspection
\begin{equation}\label{4}
\partial_t \tilde{F}_n(x,t) = \tfrac{1}{2}\Delta_x \tilde{F}_n(x,t)           
\end{equation} 
for $x \in (\mathbb{W}^-_n)^\circ$. We consider the boundary condition (\ref{2.13}) with directional derivative taken from
$(\mathbb{W}^-_n)^\circ$. One has 
\begin{eqnarray}\label{5}
&&\hspace{-40pt}(\partial_{\ell+1} - \tau\partial_\ell )\tilde{F}(x,t)\big |_{x_\ell = x_{\ell+1}} =
 \tau^{n(n-1)/2} \int_{\mathcal{C}} \mathrm{d}z_1 ... \mathrm{d}z_n
\prod_{j=1}^n \frac{1}{z_j}\cdot \frac{\tau -1}{z_j + (1-\tau)}\, \mathrm{e}^{ \frac{1}{2}tz_j^2}\nonumber\\
&&\hspace{0pt}
\times\Big(\mathop{\prod^n_{j=1}}_{j\neq\ell,\ell+1}\mathrm{e}^{x_jz_j}\Big) (z_{\ell +1} - z_{\ell})\,\mathrm{e}^{x_\ell(z_\ell + z_{\ell+1})}
\mathop{\prod_{1\leq A < B \leq n}}_{(A,B) \neq (\ell,\ell+1)} \frac{z_B - z_A}{z_B - \tau z_A}\,.
\end{eqnarray} 
The integrand has no poles in the strip bordered by $\mathcal{C}_\ell$ and $\mathcal{C}_{\ell +1}$.
Hence $\mathcal{C}_\ell$ can be moved on top of  $\mathcal{C}_{\ell+1}$. The integrand is odd under interchanging $z_\ell$
and $z_{\ell+ 1}$ and the right hand side of (\ref{5}) vanishes.\medskip\\
\textit{(ii) initial conditions.} We have to show that $\lim_{t \to 0}\tilde{F}_n(x,t) = F_n(x)$. Note that the integrand in
(\ref{3}) has an integrable bound at infinity uniformly in $t$ and hence one can set $t =0$. We define the sector $S_\ell$ by
\begin{equation}\label{6}
x_n < ... < x_{\ell+1} < 0 < x_\ell < ...< x_1           
\end{equation} 
with $\ell = 1,...,n$. Then
\begin{equation}\label{7}
F_n(x) |_{S_\ell} = \exp{\big[-(1-\tau) \sum_{j=1}^\ell \tau^{j-1}x_j\big]}          
\end{equation} 
and $\tilde{F}_n(x,0)$ will be computed for the sector $S_\ell$. Since $0<x_\ell<...<x_1$, $\exp(x_jz_j)$
 decays exponentially as $\Re{z_j} \to -\infty$, $j = 1,...,\ell$, and the contours $\mathcal{C}_1,...,\mathcal{C}_\ell$
can be deformed to circles around $z = -(1-\tau)$, maintaining the nesting condition.
Correspondingly, since $x_n<...<x_{\ell+1} < 0$, the contours $\mathcal{C}_{\ell+1},...,\mathcal{C}_n$
can be deformed to circles around $z = 0$, maintaining the nesting condition.  

We integrate first over $z_1$. Then on $S_\ell$, denoting the deformed contours by $\tilde{\mathcal{C}}_j$,
\begin{eqnarray}\label{8}
&&\hspace{-40pt}\tilde{F}_n(x,0) 
\\
&&\hspace{-10pt}
=  \tau^{n(n-1)/2} \int_{\mathcal{\tilde{C}}} \mathrm{d}z_2 ... \mathrm{d}z_n\int_{\mathcal{\tilde{C}}_1} \mathrm{d}z_1
\prod_{j=1}^n \frac{1}{z_j}\cdot \frac{\tau -1}{z_j + (1-\tau)}\, \mathrm{e}^{x_jz_j }
\prod_{1\leq A < B \leq n} \frac{z_B - z_A}{z_B - \tau z_A}\nonumber\\
&&\hspace{-10pt}=\mathrm{e}^{-(1 - \tau)x_1}
 \tau^{n(n-1)/2} \int_{\mathcal{\tilde{C}}} \mathrm{d}z_2 ... \mathrm{d}z_n
 \prod_{j=2}^n \frac{1}{z_j} \cdot\frac{\tau -1}{z_j + \tau(1-\tau)}\, \mathrm{e}^{x_jz_j }
\prod_{2\leq A < B \leq n} \frac{z_B - z_A}{z_B - \tau z_A}\,.\nonumber
\end{eqnarray} 
Iterating the integrations over $z_2,...,z_\ell$ yields
\begin{eqnarray}\label{8a}
&&\hspace{-54pt}\tilde{F}_n(x,0) =  \tau^{n(n-1)/2}
\tau^{-1}...\,\tau^{-(\ell -1)} \mathrm{e}^{-(1-\tau)x_1} ...\, 
\mathrm{e}^{-\tau^{(\ell-1)}(1-\tau)x_\ell} 
\nonumber\\
&&\hspace{0pt}\times\int_{\mathcal{\tilde{C}}} \mathrm{d}z_{\ell+1} ...\, \mathrm{d}z_n\prod_{j=\ell+1}^n \frac{1}{z_j} 
\cdot\frac{\tau -1}{z_j + \tau^\ell(1-\tau)}\, \mathrm{e}^{x_jz_j }
\prod_{\ell +1\leq A < B \leq n} \frac{z_B - z_A}{z_B - \tau z_A}\,.
\end{eqnarray} 
Next we integrate successively over $z_n$ up to $z_{\ell + 1}$. Abbreviating
\begin{equation}\label{9a}
\varpi = \sum_{j=1}^{\ell-1}j + \sum_{j=\ell}^{n-2}(n-j-1) + \ell(n-\ell)\,,
\end{equation}
one obtains
\begin{equation}\label{9}
\tilde{F}_n(x,0) = \tau^{n(n-1)/2}
 \mathrm{e}^{-(1-\tau)x_1} ...\, 
\mathrm{e}^{-\tau^{(\ell-1)}(1-\tau)x_\ell} \tau^{-\varpi} = F_n(x)\,.\medskip
\end{equation} 
\textit{(iii) uniqueness.} To show that necessarily $F_n(x,t) = \tilde{F}_n(x,t)$, we adopt an argument of Warren in a similar context \cite{War07}.
Let us consider $ \tilde{F}_n(x(t),T +\epsilon - t)$, $0\leq t\leq T$. By Ito's formula
\begin{eqnarray}\label{10}
&&\hspace{-50pt}\mathrm{d} \tilde{F}_n(x(t),T +\epsilon - t) = \big( -\partial_t \tilde{F}_n(x(t),T +\epsilon - t) +\tfrac{1}{2}\Delta_x 
\tilde{F}_n(x(t),T +\epsilon - t)\big)\mathrm{d}t \nonumber\\
&&\hspace{11pt}+ \sum _{j=1}^n \partial_{x_j} \tilde{F}_n(x(t),T +\epsilon - t)\big(\mathrm{d}B_j(t) -p\mathrm{d}\Lambda^{(j,j+1)}(t)+q \mathrm{d} \Lambda^{(j-1,j)}(t)\big)\,.
\end{eqnarray} 
The $\mathrm{d}t$ term vanishes because of (\ref{4}) and the Skorokhod term vanishes, because $\tilde{F}_n$ satisfies the boundary condition (\ref{2.13}). Hence
\begin{equation}\label{11}
\mathbb{E}_x\big(\tilde{F}_n(x(\epsilon),T)\big) = \mathbb{E}_x\big(\tilde{F}_n(x(T),\epsilon)\big)
\end{equation} 
and, taking the limit $\epsilon \to 0$,
\begin{equation}\label{12}
\tilde{F}_n(x,T) = \mathbb{E}_x\big(\tilde{F}_n(x(0),T)\big) = \mathbb{E}_x\big(F_n(x(T))\big)= 
F_n(x,T)\,,
\end{equation} 
as claimed. 
\end{proof}

\section{From moments to a Fredholm determinant}
\setcounter{equation}{0}
At the level of multi-point generating functions it is difficult to proceed any further and we concentrate on a single point by setting $x_j = u$ for all $j=1,...,n$. Then
\begin{equation}\label{13}
\mathbb{E}\big(\tau^{nN(u,t)}\big) = (-1)^n \tau^{n(n-1)/2} \int_{\mathcal{C}} \mathrm{d}z_1 ... \mathrm{d}z_n
\prod_{j=1}^n \frac{1}{z_j} f(z_j;u,t)
\prod_{1\leq A < B \leq n} \frac{z_B - z_A}{z_B - \tau z_A}
\end{equation}
with
\begin{equation}\label{14}
f(z;u,t) =
\frac{1 - \tau}{z + (1-\tau)}\, \mathrm{e}^{uz+ \frac{1}{2}tz^2} \,.
\end{equation}
The goal of this section is to obtain a Fredholm determinant for the $\tau$-deformed generating function 
of $\zeta \tau^{N(u,t)}$, i.e.
\begin{equation}\label{16}
\mathbb{E}\big(e_\tau(\zeta \tau^{N(u,t)})\big) = \mathbb{E}\big(\frac{1}{(\zeta \tau^{N(u,t)};\tau)_\infty}\big)\,.
\end{equation}
The required definitions for $\tau$-deformed objects are well summarized in Appendix A of
\cite{BCS14}. We will use the method developed in \cite{BC13} with the adaptation  \cite{BCF13} in case the defining contour is 
unbounded. We follow rather closely   \cite{BCF13} and do not repeat the full details.


\begin{figure}
\begin{center}
\psfrag{zj-1}[lb]{$z_{j-1}$}
\psfrag{-(1-t)}[lb]{$-(1-\tau)$}
\psfrag{contour j}[lb]{\hspace{-5pt}contour $j$}
\includegraphics[height=5cm]{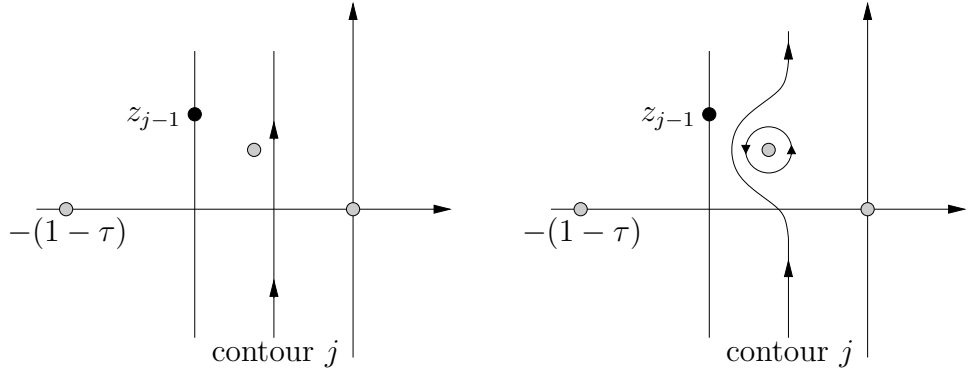}
\caption{A single move in unnesting the contours. Displayed is only the move of contour $j$ across the singularity generated by a fixed point on contour $j-1$.}
\end{center}
\end{figure}


The first step is to remove the nesting constraint by moving the contours. In Fig. 1 we display a single move.
$z_{j-1}$ is fixed and the integration is over $z_j$. The singularity for $z_j$ is at $\tau z_{j-1}$. We deform the $z_j$ contour
across the singularity and thereby pick up a pole contribution, which is evaluated by the residue theorem and identical to the expression in \cite{BCF13}, Proposition 4.11. After all contours have been moved, no singularities remain and one can further deform to a  common contour, which is is denoted by $\mathcal{C}_0 = \{ -\delta +\mathrm{i}\varphi, \varphi\in\mathbb{R}\}$ with $0 < \delta < 1 - \tau$. 
The resulting combinatorial structure is identical to the one  \cite{BCF13}.  The integrand for the $n$-th 
 moment in \cite{BCF13}, Lemma 4.10, is replaced by the expression from (\ref{13}). 

Using the $\tau$-binomial theorem and rearranging terms, one arrives at a Fredholm 
determinant of the $\tau$-deformed generating function (\ref{16}). 
\begin{pro}\label{prop5}
There exists a positive constant $C\geq 1$ s.t. for all $|\zeta|<C^{-1}$, 
\begin{equation}\label{14a}
 \mathbb{E} \left[ \frac{1}{(\zeta \tau^{N(u,t/\gamma)};\tau)_\infty} \right]
 =
 \det(1+K)_{L^2(\mathbb{Z}_{>0}\times \mathcal{C}_0)}\,, 
\end{equation}
where the kernel $K$ is given by 
\begin{equation}\label{15a}
K(n_1,w_1;n_2,w_2)
=
\frac{\zeta^{n_1} f(w_1;u,t)f(\tau w_2;u,t) \cdots f(\tau^{n_1-1}w_1;u,t)}{\tau^{n_1-1} w_1-w_2}\,.
\end{equation}
\end{pro}
\begin{proof}
Our contour $\mathcal{C}_0$ differs  from the one in  \cite{BCF13}. But clearly, for all $n \geq1$,
\begin{equation}\label{16a}
|\tau^{n} w_1-w_2| \geq (1 - \tau) \delta\,.
\end{equation}
Also $f(w;u,t)$ is bounded and decays as a Gaussian for large
$|w|$, which ensure the convergence of the Fredholm expanded determinant
for small enough $|\zeta|$. 
\end{proof}\\

The above Fredholm determinant is not yet suitable for asymptotics and one has to replace the sum over $n$
by a contour integral.
We introduce $g$  as the solution of $f(z,t) = g(z,t)/g(\tau z,t)$. Then 
\begin{equation}\label{15}
g(z,t) = \exp\big[\big( u (1-\tau)^{-1}z + \frac{1}{2}\gamma t (1-\tau)^{-2}z^2\big)\big] \frac{1}{(-(1-\tau)^{-1}z; \tau)_{\infty}}
\end{equation}
with $\gamma = q -p$. Clearly the natural units are $\gamma t$ and $w= (1-\tau)^{-1}z$ and we set
\begin{equation}\label{16b}
g((1-\tau)w,t/\gamma) = \tilde{g}(w,t) = \mathrm{e}^{uw  + \frac{1}{2}t w^2} \frac{1}{(-w; \tau)_{\infty}}. 
\end{equation}
\begin{theorem}\label{th6}
 Let $\zeta \in \mathbb{C}\setminus\mathbb{R}_+$. Then 
\begin{equation}\label{17}
\mathbb{E}\big( \frac{1}{(\zeta \tau^{N(u,t/\gamma)};\tau)_\infty} \big)= \det(1+K_\zeta)\,.
\end{equation}
The kernel $K_\zeta$ is given by
\begin{equation}\label{18}
K_\zeta(w,w') = \int _{\mathcal{C}_w} \mathrm{d}s \Gamma(-s)\Gamma(1+ s)(-\zeta)^s
\frac{\tilde{g}(w,t)}{\tilde{g}(\tau^sw,t)} \,\frac{1}{\tau^sw - w'}\,.
\end{equation}
Here $w,w'\in\mathcal{C}_0$ and the $s$-contour $\mathcal{C}_w$
is explained below. The kernel $K_\zeta(-\delta +\mathrm{i}\varphi, -\delta +\mathrm{i}\varphi')$ depends smoothly on
$\varphi,\varphi'$ and satisfies the bound
\begin{equation}\label{19}
|K_\zeta(-\delta +\mathrm{i}\varphi, -\delta +\mathrm{i}\varphi') |\leq c_0 \mathrm{e}^{-c_1\varphi^2}
\end{equation}
with a suitable choice of $c_0,c_1 > 0$.
\end{theorem}
\begin{proof}
The $\mathcal{C}_w$-contour is shown in Fig. 2. The contour is reflection symmetric relative to the real axis and piecewise linear
starting from $\tfrac{1}{2}$, to $\tfrac{1}{2} +\mathrm{i} d$, to $R + \mathrm{i} d$, to $R +\mathrm{i} \infty$, $d> 0, R \geq \tfrac{1}{2}$. The parameters $d,R$ depend on $w$. For small $|\varphi|$ we set $R = 1/2$, while for
large $|\varphi|$ we choose $d = c_4/|\varphi|$ and $|\varphi|\tau^R = \delta/2$.
 \medskip\\
\textit{Gaussian bound}. For $w = -\delta + \mathrm{i}\varphi$ and $s \in \mathcal{C}_w$ it holds
\begin{equation}\label{20}
\big|\frac{\tilde{g}(w,t)}{\tilde{g}(\tau^sw,t)} \big| = \big| \exp \big( u(1 - \tau^s)w + \frac{1}{2}t w^2 (1 - \tau^{2s})\big) \big|
\cdot\big|
\frac{(-\tau^sw;\tau)_\infty}{(-w;\tau)_\infty} \big| \leq c_0 \mathrm{e}^{-c_1\varphi^2}\,.
\end{equation}
The first factor is estimated as $|\exp(\cdot)| \leq \exp\big(b_1 + b_2|\varphi|  - b_3\varphi^2\big)$, where 
$b_3 = 1- (\Re \tau^s)^2 + (\Im \tau^s)^2 \geq 1-\tau$.  Therefore the $\varphi^2$ term dominates the linear
term and provides the Gaussian bound.

The second factor is written as 
\begin{equation}\label{21}
\frac{(-\tau^sw;\tau)_\infty}{(-w;\tau)_\infty} = \frac{1+w\tau^s}{1+w}\cdot\prod_{n=1}^\infty\big(1 + (\tau^s - 1) \frac{1}{1 + \tau^{-n}w^{-1}}\big)\,.
\end{equation}
Since $|w+\tau^{-n}| \geq \tau^{-n}|1- \tau(1 -\tau)| \geq \tfrac{1}{2}\tau^{-n}$, one arrives at
\begin{equation}\label{22}
\big| (\tau^s - 1) \frac{1}{1 + \tau^{-n}w^{-1}}\big| \leq |(\tau^s - 1)| |w|2 \tau^{n}\,,
\end{equation}
which implies that the second product converges uniformly in $s \in
\mathcal{C}_w$ with a bound proportional to $|w|$.\medskip\\
\textit{Integration along $\mathcal{C}_w$}. We will show that
\begin{equation}\label{23}
\int _{\mathcal{C}_w} \mathrm{d}s\big| \Gamma(-s)\Gamma(1+ s)(-\zeta)^s
\frac{1}{\tau^sw - w'} \big|\leq c_3 (1+ \log(1+|\varphi|))
\end{equation}
with $c_3$ depending only on $\zeta$.

Considering the third factor,  the contour $\mathcal{C}_w$ has been constructed such that
$|\tau^s w' - w| \geq a_0 > 0$ for all $w,w'$ and $s\in \mathcal{C}_w$,  uniformly in $s$.

For the second factor we set $-\zeta = |\zeta|\mathrm{e}^{\mathrm{i} \theta}$, which implies $|\theta| < \pi$
by assumption. The contributing part of the contour  integration is $\{s = R + \mathrm{i} r, d \leq r < \infty\}$ and its mirror image. 
Along this part it holds
\begin{equation}\label{24}
|(-\zeta)^s| \leq |\zeta|^R \mathrm{e}^{|\theta| r}\,.
\end{equation}

For the first factor we use the identity $\Gamma(-s)\Gamma(s) = \pi/ \sin(\pi s)$. Inserting the previous bounds
 \begin{equation}\label{25}
\int _{\mathcal{C}_w} \mathrm{d}s\big| \Gamma(-s)\Gamma(1+ s)(-\zeta)^s
\frac{1}{\tau^sw - w'} \big|\leq c_5(1 + |\zeta|^R (\pi- |\theta|)^{-1})\,.
\end{equation}
Since $R \propto \log |\varphi|$ for large $|\varphi |$, the bound (\ref{23}) follows. 

The logarithmic divergence (\ref{23}) and the linear bound in (\ref{22}) can 
be absorbed into the Gaussian bound (\ref{20}) and the bound (\ref{19}) is established. 
\end{proof}

\begin{figure}
\begin{center}
\psfrag{Cw}[lb]{$\mathcal{C}_w$}
\psfrag{2d}[lb]{$2d$}
\psfrag{R}[lb]{$R$}
\psfrag{1}[lb]{$1$}
\includegraphics[height=5cm]{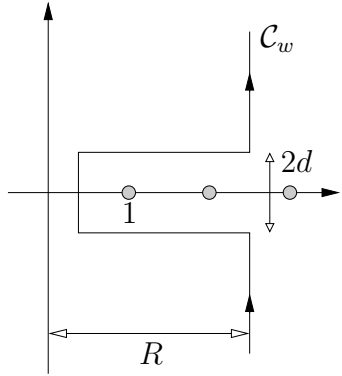}
\caption{Complex $s$-plane and the integration contour $\mathcal{C}_w$. Poles of the integrand are located at the positive integers.}
\end{center}
\end{figure}


\section{Formal asymptotics}\label{sec5}
\setcounter{equation}{0}
To obtain the long time asymptotics of $N(u,t/\gamma)$ requires a steepest decent analysis of the kernel $K_\zeta$
of (\ref{18}). Here we only identify the saddle point and its expansion close to the saddle. Thereby the GUE asymptotics becomes visible.  For a complete proof a more detailed analysis of the steepest decent path would have to be carried out. 

One first has to figure out the law of large numbers for $N(u,t)$. The quick approach is to use the ASEP, $\tfrac{1}{2} < q\leq1$, with step initial conditions. On the macroscopic scale the density, $\rho$,  is governed by
\begin{equation}\label{5.11}
 \partial_t \rho - \gamma \partial_u (\rho - \rho^2) = 0\,.
\end{equation}
To match with the Brownian motions, one has to shift $\rho$ to $\tilde{\rho}(u,t) = \rho(u+\gamma t,t)$. Then $\tilde{\rho}$ satisfies
\begin{equation}\label{5.12}
 \partial_t \tilde{\rho} + \gamma\partial_u \tilde{\rho}^2 = 0\,.
\end{equation}
The solution with initial data $\tilde{\rho}(u,0) = \theta(u)$ reads 
$\tilde{\rho}(u,t/\gamma) = u/2t$ for $0 \leq u \leq 2t$. We scale $u = at$ with $a > 0$ and eventually $t \to\infty$. Then
to leading order
\begin{equation}\label{5.13}
 N(at,t/\gamma) = \tfrac{1}{4}a^2 t\,,\,\,0 \leq a \leq 2\,,\quad N(at,t/\gamma) = (a-1)t\,,\,\, 2 \leq a\,.
\end{equation}
For $a > 2$ one expects to have Gaussian fluctuations of size $\sqrt{t}$, while for $a <2$ the fluctuations should be KPZ like of size $t^{1/3}$. In the following we restrict to $0< a < 2$. The same law of large numbers can be obtained from
Eq. (\ref{1.10}) for
the macroscopic stretch $\ell$, by noting that $P(\ell) = \ell^{-1}$ for point interactions.
 
 We substitute $z = \tau^s w$, $s \log \tau = \log z - \log w$ and set
\begin{equation}\label{5.1}
-\zeta = \tau^{-\frac{1}{4}a^2 t + rt^{1/3}}\,,\quad (-\zeta)^s = \exp\big((-\tfrac{1}{4}a^2 t + rt^{1/3})(\log z - \log w)\big)\,.
\end{equation}
Inserting on the left hand side of (\ref{17}), it follows, see \cite{BCS14}, Lemma 4.1.39, that
\begin{equation}\label{5.2}
\lim_{t \to \infty} \mathbb{E}\big( \frac{1}{(\zeta \tau^{N(u,t/\gamma)};\tau)_\infty} \big)= \lim_{t \to \infty}
\mathbb{P}\big(t^{-1/3}(N(u,t/\gamma) -\tfrac{1}{4}a^2t)\geq -r\big)\,.
\end{equation}

Thus we have to study the corresponding limit on the right hand side of (\ref{18}). 
In the new coordinates the kernel reads
\begin{eqnarray}\label{5.2a}
&&\hspace{-42pt}K_\zeta(w,w') =  \frac{1}{\log \tau}\int _{\mathcal{C}_w}dz\frac{1}{z}\cdot \frac{1}{z - w'}\cdot
\frac{{\pi}}{\sin(\pi(\log \tau)^{-1}
(\log w - \log z))} \nonumber\\
&&\hspace{20pt}\exp\big(t(G(z) - G(w) ) + r t^{1/3} (\log z - \log w)\big) \cdot \frac{(-z;\tau)_\infty}{(-w;\tau)_\infty}\,,
\end{eqnarray}
where
\begin{equation}\label{5.3}
G(z) = -\tfrac{1}{2}z^2 -az -\tfrac{1}{4}a^2 \log z\,.
\end{equation}
Note that 
\begin{equation}\label{5.4}
G'(z) = - \frac{1}{z}(z + \tfrac{1}{2}a)^2\,,\quad z_{\mathrm{c}} = -\tfrac{1}{2} a\,,\quad G''(z_{\mathrm{c}}) =0\,,\quad G'''(z_{\mathrm{c}}) = -\frac{2}{z_{\mathrm{c}}}\,.
\end{equation}
We expand the kernel at the saddle by setting $z = z_{\mathrm{c}}(1 + t^{-1/3}\tilde{z})$,
$w = z_{\mathrm{c}}(1+ t^{-1/3} \tilde{w})$, $w' = z_{\mathrm{c}}(1+ t^{-1/3} \tilde{w}')$.  Then, in the limit $t \to \infty$,
\begin{eqnarray}\label{5.5}
\hspace{0pt}\frac{1}{z} \mathrm{d}z \simeq t^{-1/3} \mathrm{d}\tilde{z}\,,\quad\frac{1}{z - w'} &=& \frac{z_{\mathrm{c}}t^{1/3}}{\tilde{z}- \tilde{w}'}\,,\\
\hspace{0pt} \frac{1}{\log \tau} \cdot \frac{{\pi}}{\sin(\pi(\log \tau)^{-1}(\log w - \log z))} 
 &\simeq& \frac{t^{1/3}}{\tilde{w} -\tilde{z}}\,,
\\
\hspace{0pt} t \big(G(z_{\mathrm{c}}(1+ t^{-1/3} \tilde{z})) - G(z_{\mathrm{c}}(1+ t^{-1/3} \tilde{w}))\big)
&\simeq&  - \frac{1}{3}z_{\mathrm{c}}^2\big( \tilde{z}^3  -   \tilde{w}^3\big)\,, \\
\hspace{0pt}  r t^{1/3} \big( \log (z_{\mathrm{c}}(1 + t^{-1/3} \tilde{z})) -   \log (z_{\mathrm{c}}(1 + t^{-1/3} \tilde{w}))\big)  
&\simeq& r(\tilde{z} - \tilde{w})\,,\\
\hspace{0pt}  \frac{(-z;\tau)_\infty}{(-w;\tau)_\infty}& \simeq& 1\,.
\end{eqnarray}
There is an extra factor $(z_{\mathrm{c}}t^{1/3})^{-1}$ from the volume element due to the  change in  $w,w'$.

We substitute $\tilde{z},\tilde{w},\tilde{w}'$ by $(a/2)^{-2/3}z,(a/2)^{-2/3}w,(a/2)^{-2/3}w'$ 
and thereby arrive at the limiting kernel
\begin{equation}\label{5.6}
K_r(w,w') =  \int \mathrm{d} z \exp\big(-\tfrac{1}{3}z^3 + \tfrac{1}{3}w^3 + 
(a/2)^{-2/3}r (z-w) \big)\frac{1}{w-z}\cdot \frac{1}{z - w'}\,.
\end{equation}
The $w$ contour is now given by two rays departing at 1 at angles $\pm \pi/3$, oriented with increasing imaginary part,
and the $z$ contour is given by two infinite rays starting at $0$ at angles $\pm 2\pi/3$, oriented with decreasing imaginary part. The Fredholm determinant with this kernel is identical to the Fredholm determinant of the Airy kernel, see \cite{TW09}
Lemma 8.6. Hence one concludes that
\begin{equation}
\lim_{t \to \infty}\mathbb{P}\big(t^{-1/3}(N(u,t/\gamma) -\tfrac{1}{4}a^2t)\geq -(a/2)^{2/3}r\big) = F_{\mathrm{GUE}}(r)
\end{equation}
with $F_{\mathrm{GUE}}(r) = \mathbb{P}(\xi_\mathrm{{GUE}} \leq r)$, under the assumption that the contribution from the 
remainder of the steepest decent path vanishes as $t \to \infty$.

\section{The Bethe ansatz transition probability}\label{sec6}
\setcounter{equation}{0}
The goal of this section is to establish that the dynamics with point interactions has a ``smooth" transition probability, as used in 
Section \ref{sec2} for the proof of duality. While there should be a more abstract approach, we will use the 
Bethe ansatz construction of the transition probability, as pioneered by Tracy and Widom \cite{TW08,TW09}
in the context of  the ASEP. To make the comparison transparent,
we follow closely their notation, which in part deviates from earlier notations. The particle process is denoted by 
$x(t) \in \mathbb{W}_N^+$ with initial condition $x(0) = y$. As explained before $x(t)$ is the semi-martingale
determined by
\begin{equation}\label{6.1}
x_j(t) = y_j + B_j(t) - p \Lambda^{(j,j+1)}(t) +q \Lambda^{(j-1,j)}(t)\,, 
\end{equation}
$t\geq 0$, $j = 1,...,N$. By definition $\Lambda^{(0,1)}(t) = 0 =  \Lambda^{(N,N+1)}(t)$, where
\begin{equation}\label{6.2}
\Lambda^{(j,j+1)}(\cdot) = L^{x_{j+1} -x_j}(\cdot,0) 
\end{equation}
 is the right-sided local time accumulated at the origin by the nonnegative martingale $x_{j+1}(\cdot) -x_j(\cdot)$.
 
 Let $f:  \mathbb{W}_N^+ \to \mathbb{R}$ be a $C^2$-function and define
 \begin{equation}\label{6.3a}
f(y,t) = \mathbb{E}_y\big(f(x(t)\big)
 \end{equation} 
with $\mathbb{E}_y$ denoting expectation of the $x(t)$ process of (\ref{6.1}) starting at $y \in \mathbb{W}_N^+$. 
As to be shown, $f$ satisfies the backwards equation
\begin{equation}\label{6.4a}
\partial_t f = \tfrac{1}{2} \Delta_y f
\end{equation} 
for $ y \in (\mathbb{W}_N^+)^\circ$ and
\begin{equation}\label{6.5a}
(p\partial_j - q\partial_{j+1})f \big |_{y_j = y_{j+1}} = 0\,,
\end{equation} 
the directional derivative being taken from the interior of $\mathbb{W}_N^+$.

Let us define the standard decomposition
\begin{equation}\label{6.3}
\mathbb{P}\big(x(t) \in \mathrm{d}x\big|x(0) = y\big) = P_y(x,t) \mathrm{d}x + P_y^{\mathrm{sing}}(\mathrm{d}x,t)\,. 
\end{equation}
In spirit $P_y(x,t) $ should be the solution to the backwards equation. We follow Bethe \cite{B35}
and start from  an ansatz for the solution of (\ref{6.4a}), (\ref{6.5a}) given by
\begin{equation}\label{6.6f}
 Q_y(x,t) = \sum_{\sigma \in S_N} \int_{\Gamma_a}  \mathrm{d}z_1\cdots\int_{\Gamma_a} \mathrm{d}z_N\,A_\sigma(\underline{z}) \prod_{j=1}^N\mathrm{e}^{z_{\sigma(j)}(x_j - y_{\sigma(j)})}\mathrm{e}^{\frac{1}{2}z_j^2t}
 = \sum_{\sigma \in S_N} I_\sigma(y;x,t)\,,
  \end{equation}
where the sum is over all permutations $\sigma$ of order $N$. The Gaussian 
factor ensures that Eq. (\ref{6.4a}) is satisfied. The expansion coefficients $A_\sigma$ are determined  
through the boundary condition (\ref{6.5a}). We define
the ratio of scattering amplitudes 
\begin{equation}\label{6.4}
S(z_\alpha,z_\beta) = - \frac{\tau z_\alpha -  z_\beta}{\tau z_\beta - z_\alpha}
\end{equation}
for wave numbers $z_\alpha, z_\beta \in \mathbb{C}$.
The expansion coefficient $A_\sigma$ can be written as
\begin{equation}\label{6.5}
A_\sigma(\underline{z}) = \prod_{\{\alpha,\beta\} \in \mathrm{In}(\sigma)}S(z_\alpha,z_\beta)\,.
\end{equation}
$\underline{z}$ stands for $(z_1,...,z_N)$. $\mathrm{In}(\sigma)$ denotes the set of all inversions in $\sigma$, where an inversion in $\sigma$ means an ordered pair $\{\sigma(i),\sigma(j)\}$ such that $i < j$ and $\sigma(i) > \sigma(j)$.  The contour of integration is $\Gamma_a = \{a +\mathrm{i}\varphi, \varphi \in \mathbb{R}\}$  with positive orientation. 
\begin{theorem}\label{th7}
Let $0 <\tau < 1$ and $a >0$. For $t>0$ and every $y \in \mathbb{W}_N^+$ the transition probability for $x(t)$ is absolutely continuous, $\mathbb{P}\big(x(t) \in \mathrm{d}x\big|x(0) = y\big) =  P_y(x,t) \mathrm{d}x$. Its density has a continuous version on $\mathbb{W}_N^+$ given by  
\begin{equation}\label{6.6}
P_y(x,t) = Q_y(x,t) \quad\mathrm{a.s.}\,. \smallskip
\end{equation}
\end{theorem}
For $1 < \tau < \infty$, Eq.  (\ref{6.6}) still holds, but one has to impose $a < 0$. The limiting cases $\tau = 1$ and $\tau \to 0$ will be discussed below.

We first investigate properties of  $Q_y(x,t)$ and set 
\begin{equation}\label{6.6a}
Q_y(f,t) =  \int_{\mathbb{W}_N^+} \mathrm{d}x Q_y(x,t)f(x) \,,\quad I_\sigma(y;f,t)  = \int_{\mathbb{W}_N^+} \mathrm{d}x  I_\sigma(y;x,t)f(x)
\end{equation}
with properties of the test function $f$ to be specified later on.
\begin{lemma}\label{le8}
Let $y \in  \mathbb{W}_N^+$ and let $f \in \mathcal{D}_\epsilon$, which consists of smooth functions with compact support contained in 
$ \mathbb{W}_{N,\epsilon}^+ = \{x \in  \mathbb{W}_N^+|\, |x_{j+1} - x_j| \geq \epsilon, \,\mathrm{all}\,j\}$. Then for $x \in  (\mathbb{W}_N^{+})^\circ$ it holds 
\begin{equation}\label{6.6b}
\partial_t Q_y(x,t) = \tfrac{1}{2} \Delta_y Q_y(x,t)\,, 
\end{equation}
\begin{equation}\label{6.6bb}
(\tau\partial_j -\partial_{j+1})Q_y(x,t)\big|_{y_j = y_{j+1}} =0\,. 
\end{equation}
For $\sigma = \mathrm{id}$, the identity permutation,
\begin{equation}\label{6.6c}
\lim_{t \to 0}I_\mathrm{id}(y;f,t)  = f(y)
\end{equation}
and for $\sigma \neq \mathrm{id}$
\begin{equation}\label{6.6d}
\lim_{t \to 0}I_\sigma(y;f,t) = 0\,. 
\end{equation}
\end{lemma}

We illustrate the method by means $N = 2$, for which
\begin{eqnarray}\label{6.10}
&&\hspace{-26pt}Q_y(x,t) = \int_{\Gamma_a} \mathrm{d}z_1 \int_{\Gamma_a} \mathrm{d}z_2\big( \mathrm{e}^{z_1(x_1-y_1)
+z_2(x_2 - y_2)} - \frac{\tau z_2 -  z_1}{\tau z_1 -  z_2} \mathrm{e}^{z_2(x_1-y_2)
+z_1(x_2 - y_1)}\big) \,\mathrm{e}^{\frac{1}{2}(z_1^2 +z_2^2)t} \nonumber\\
&&\hspace{17pt}= I_{12}(y;x,t) + I_{21}(y;x,t)\,.
\end{eqnarray}
The validity of Eq. (\ref{6.6b}) is easily checked. For the boundary condition we note
\begin{eqnarray}\label{6.10a}
&&\hspace{-36pt}(\tau \partial_1 - \partial_2) Q_y(x,t)\big |_{y_1 = y_{2}} =\int_{\Gamma_a} \mathrm{d}z_1 \int_{\Gamma_a} \mathrm{d}z_2
 \big( (-\tau z_1 + z_2)\mathrm{e}^{z_1(x_1-y_1)
+z_2(x_2 - y_1)} \nonumber \\
&&\hspace{-16pt}
-  \frac{\tau z_2 -  z_1}{\tau z_1 -  z_2}(-\tau z_1 +z_2)\mathrm{e}^{z_2(x_1-y_1)
+z_1(x_2 - y_1)} \big) \mathrm{e}^{\frac{1}{2}(z_1^2 +z_2^2)t} 
=0\,.
\end{eqnarray}
Clearly, $I_{12}(y;x,t)$ satisfies (\ref{6.6c}). Thus we still have show that $\lim_{t \to 0}I_{21}(y;x,t)$ vanishes for  $y \in \mathbb{W}_2^+$ and $x \in (\mathbb{W}_2^{+})^\circ$.
For this purpose, we introduce a new variable, $z_0$,
by $z_0 = z_1 +z_2$ and substitute $z_2$ by $z_0$. Then 
\begin{eqnarray}\label{6.11}
&&\hspace{-43pt} I_{21}(y;x,t) = - \int_{\Gamma_{2a}} \mathrm{d}z _0\int_{\Gamma_a} \mathrm{d}z_1 \frac{\tau (z_0 - z_1) -z_1 }{\tau z_1 -(z_0-z_1)} \,\mathrm{e}^{(z_0 - z_1)(x_1-y_2)
+z_1(x_2 - y_1)} \mathrm{e}^{\frac{1}{2}(z_1^2 + (z_0-z_1)^2)t}\nonumber\\
&&\hspace{-30pt}=  \int_{\Gamma_{2a}} \mathrm{d}z_0 \int_{\Gamma_a} \mathrm{d}z_1 \frac{z_1 - \tau(1+\tau)^{-1}z_0}{z_1 - (1+\tau)^{-1}z_0} \,\mathrm{e}^{z_1(x_2 - x_1 +y_2 - y_1)+z_0(x_1-y_2) 
} \mathrm{e}^{\frac{1}{2}(z_1^2 + (z_0-z_1)^2)t}\,.
\end{eqnarray}
The pole of $z_1$ is at $z_1 = (1+\tau)^{-1}z_0$ and hence to the right of $\Gamma_a$. Under our assumptions
one has $x_2 -x_1 +y_2 - y_1 > 0$. For the limit $t \to 0$ the following distributional identities will be used.\medskip\\
For $a\in \mathbb{R}$, $b\in \mathbb{C}$ with $\Re b<a$ it holds
\begin{equation}\label{6.12}
\int_{\Gamma_a}\mathrm{d}z \frac{\mathrm{e}^{(z-b)u}}{z-b} = \theta(u)
\end{equation}
and  with $\Re b> a$
\begin{equation}\label{6.13}
\int_{\Gamma_a}\mathrm{d}z \frac{\mathrm{e}^{(z-b)u}}{z-b} = - \theta(-u)\,,
\end{equation}
where $ \theta(u) = 1$ for $u >0$ and $ \theta(u) = 0$ for $u <0$. Using (\ref{6.13}) implies $\lim_{t \to 0}I_{21}(y;x,t) = 0$
for $x \in (\mathbb{W}_2^{+})^\circ$.\\\\
\textbf{Proof of Lemma \ref{le8}}:  The properties (\ref{6.6b}) and (\ref{6.6bb}) are easily checked. Also Property (\ref{6.6c}) 
follows directly from the definition. The difficult part is (\ref{6.6d}). In fact, for the ASEP the analogue of 
$I_\sigma$ is not necessarily equal to $0$ and one has to use cancellations. In this respect the contour integral for  Brownian motions with oblique reflections  has a somewhat simpler pole structure than its lattice gas version. 

We choose subsets $A,B \subset [1,...,N-1]$, such that $A\cap B = \emptyset$, $A\cup B = [1,...,N-1]$, $0\leq |A| \leq N-2$,
$1\leq |B| \leq N-1$. For $1 \leq n < N$ we set $\sigma(n) = N$, $A = \{i_1,...,i_{n-1}\}$ and  $B = \{i_{n+1},...,i_{N}\}$.
A generic permutation then reads 
\begin{equation}\label{6.16}
\sigma =
\begin{pmatrix}
1&2&\cdots& n-1 & n &n+1& \cdots&N\\
i_1& i_2&\cdots& i_{n-1} & N&i_{n+1} &\cdots& i_N
\end{pmatrix}\,.
\end{equation}
If $n = N$, one falls back onto the case $N-1$. Thus without loss of generality one can restrict to $n< N$.

By separating the factors corresponding to the inversions $(N,j)$ with $j\in B$, the integrand of $I_\sigma$ can be written as
\begin{equation}\label{6.17}
\prod_{j\in B} S(z_N,z_j)
\prod_{\{\alpha,\beta\} \in \mathrm{In}(\sigma), \,\alpha \neq N} S(z_\alpha,z_\beta) 
 \prod_{j=1}^N \mathrm{e}^{z_j(x_{\sigma^{-1}(j)} - y_j)} \mathrm{e}^{\frac{1}{2}z_j^2t}\,.
\end{equation}
We set
\begin{equation}\label{6.18}
z_0 = z_1 +...+z_{N}
\end{equation}
and substitute $z_N$ by $z_0$. Hence $z_0 \in \Gamma_{Na}$. The phase factor transforms to
\begin{equation}\label{6.18a}
 \prod_{j=1}^N \mathrm{e}^{z_j(x_{\sigma^{-1}(j)} - y_j)} 
 =  \prod_{j=1}^{N-1}  \mathrm{e}^{z_j(x_{\sigma^{-1}(j)} -x_{\sigma^{-1}(N)}+ y_N  - y_j)}
 \mathrm{e}^{z_0(x_{\sigma^{-1}(N)} -y_N)}\,.
\end{equation}
Since $\sigma^{-1}(N)= n$ and $ n< \sigma^{-1}(j)$ for $j\in B$, one concludes
\begin{equation}\label{6.18b}
x_{\sigma^{-1}(j)} -x_{\sigma^{-1}(N)}+ y_N  - y_j \geq \epsilon\,,\quad j\in B\,.
\end{equation}

We set $\ell = \mathrm{min}\,B$ and first integrate over $z_\ell$. Poles may arise from $S(z_N,z_j)$ $[case\,\,1]$ and $S(z_\alpha,z_\beta)$ $[case\,\,2]$.
In the first case, if $j = \ell$, the denominator reads
\begin{equation}\label{6.18c}
(1 +\tau)z_\ell - (z_0 -z_1-... - \cancel{z_\ell} -... - z_{N-1})\,. 
\end{equation}
Since $\tau <1$, the pole for the $z_\ell$ integration lies to the right of $\Gamma_a$. Furthermore, if $j \neq \ell$, the denominator reads
\begin{equation}\label{6.19}
z_\ell + (1 +\tau)z_j - (z_0 -z_1-... - \cancel{z_j}-\cancel{z_\ell} -... - z_{N-1})\,. 
\end{equation}
As before, the pole for the $z_\ell$ integration lies to the right of $\Gamma_a$.
In the second case 
a generic factor reads 
\begin{equation}\label{6.20}
S(z_\alpha,z_\beta) = - \frac{\tau z_\alpha -  z_\beta}{\tau z_\beta - z_\alpha}
\end{equation}
with $\alpha > \beta$. If $\ell = \alpha$, then $\ell >\beta$. Since $\ell = \mathrm{min}\,B$, one must have $\beta \in A$.
But then $(\ell,\beta)$ is not an inversion. Hence
$\ell = \beta$ and the pole for the $z_\ell$ integration is at $\tau^{-1}z_\alpha$ for some $\alpha \in [1,...,N-1]$ and hence to the right of $\Gamma_a$.
Thus the $z_\ell$ integration has no poles to the left of $\Gamma_a$. With this information Property (\ref{6.6d}) can be proved.
We leave the details for Appendix \ref{app.d}.
\begin{lemma}\label{le9}
For $f\in \mathcal{D}_\epsilon$ it holds
\begin{equation}\label{6.21}
\mathbb{E}_y\big(f(x(t))\big) = Q_y(f,t)\,. 
\end{equation}
\end{lemma}
\begin{proof}
Let us denote $P_y(f,t) =  \mathbb{E}_y\big(f(x(t))\big)$. We have to show that $P_y(f,t)= Q_y(f,t)$,
which corresponds to Theorem \ref{th3} upon identifying $P_y(f,t)$ with $F_n$ and  $Q_y(f,t)$ with
$\tilde{F}_n$. We have established already that $Q_y(f,t)$ satisfies the properties \textit{(i)} and \textit{(ii)} in the proof of Theorem
\ref{th3}. So we merely have to copy part \textit{(iii)} with the result
\begin{equation}\label{6.21a}
\mathbb{E}_y\big(Q_{x(\epsilon)}(f,T)\big) = \mathbb{E}_y\big(P_{x(T)}(f,\epsilon)\big)\,. 
\end{equation}
Continuously in $\epsilon$, $x(\epsilon) \to y$ and $P_{x(T)}(f,\epsilon) \to f(x(T))$. Hence
\begin{equation}\label{6.21b}
Q_y(f,T) = \mathbb{E}_y\big(f(x(T))\big) \,. 
\end{equation}
\end{proof}

The lemma tells  us that a possible singular contribution to the transition probability has to be concentrated on $\partial
\mathbb{W}_N^+$. 
\begin{lemma}\label{le10}
For $t>0$ and $y \in \mathbb{W}_N^+$,  
 \begin{equation}\label{6.22}
P_y^{\mathrm{sing}}(\mathrm{d}x,t) = 0\,. 
\end{equation}
\end{lemma}
\begin{proof}
Since  $Q_y(x,t) \geq 0$ by Lemma \ref{le9}, one has to show that $Q_y(1\!\!1,t) = 1$ with $1\!\!1(x) = 1$.
We set
\begin{equation}\label{6.22a}
g_N(u) = \int_{-\infty}^u dx_N \cdots  \int_{-\infty}^{x_2} dx_1Q_y(x,t) \,. 
\end{equation}
$g_N(u)$ is the distribution function for the $N$-th particle at fixed initial configuration $y$. Since $a > 0$, all $x$-integrals are convergent and
\begin{eqnarray}\label{6.22b}
&&\hspace{-50pt} g_N(u) = \sum_{\sigma \in S_N}\int_{\Gamma_a}  \mathrm{d}z_1\cdots\int_{\Gamma_a} \mathrm{d}z_N\,A_\sigma(\underline{z}) 
\nonumber \\ 
&&\hspace{18pt} \times\frac{1}{(z_{\sigma(1)} + \cdots + z_{\sigma(N)})\cdots(  z_{\sigma(2)}+ z_{\sigma(1)}) z_{\sigma(1)}}
\prod_{j=1}^N\mathrm{e}^{z_j(u - y_j)}\mathrm{e}^{\frac{1}{2}z_j^2t}\,.
 \end{eqnarray} 
One can rewrite 
\begin{equation}\label{6.22c}
A_{\sigma}(\underline{z})= \mathrm{sgn}\,\sigma \prod_{1 \leq i < j \leq N}\frac{qz_{\sigma(j)} -  pz_{\sigma(i)}}
{qz_j -  pz_i} \,. 
\end{equation}
To apply  the first combinatorial identity of Tracy and Widom \cite{TW08}, Section VI, one has to invert the order as $\tilde{\sigma}(j) = \sigma(N-j)$. 
Then (\ref{6.22b}) reads
\begin{eqnarray}\label{6.26}
&&\hspace{-30pt}\sum_{\tilde{\sigma}\in S_N} \mathrm{sgn}\,\tilde{\sigma} \prod_{1 \leq i < j \leq N}\frac{pz_{\tilde{\sigma}(j)} -  qz_{\tilde{\sigma}(i)}}
{qz_j -  pz_i}\, \frac{1}{(z_{\tilde{\sigma}(1)} + \cdots + z_{\tilde{\sigma}(N)})\cdots(  z_{\tilde{\sigma}(N-1)}+ z_{\tilde{\sigma}(N)}) z_{\tilde{\sigma}(N)}}\nonumber\\
&&\hspace{40pt} = q^{N(N-1)/2} \prod_{1 \leq i < j \leq N}\frac{z_j - z_i}{qz_j -  pz_i}  \prod_{j=1}^N \frac{1}{z_j}
= \prod_{1 \leq i < j \leq N}\frac{z_j - z_i}{z_j -  \tau z_i}  \prod_{j=1}^N \frac{1}{z_j}\,. 
\end{eqnarray}
In the second line we used the combinatorial identity in the limit $\xi_j = 1+ z_j$ to linear order in $z_j$.
Inserting in (\ref{6.22b}), one arrives at
\begin{equation}\label{6.22e}
g_N(u) = \int_{\Gamma_a}  \mathrm{d}z_1\cdots\int_{\Gamma_a} \mathrm{d}z_N
\prod_{1 \leq i < j \leq N}\frac{z_j -  z_i}{z_j - \tau z_i} 
\prod_{j=1}^N  \frac{1}{z_j}    \mathrm{e}^{z_j(u - y_j)}\mathrm{e}^{\frac{1}{2}z_j^2t}\,.
 \end{equation} 
We have to show that $\lim_{u \to \infty}g_N(u) = 1$.

We integrate over $z_1$. The poles for $z_1$ are at $\tau^{-1} z_j$, $z_j \in \Gamma_a$, $j= 2,...,N$,
and at $z_1 = 0$. We choose $u$ sufficiently large such that $u - y_j > 0$. Then the contour $\Gamma_a$
can be deformed to a contour $\tilde{\Gamma}_a$ plus a small positively oriented circle around 0.
$\tilde{\Gamma}_a$ coincides with $\Gamma_a$ far away from the origin and lies to the left of $z_1 = 0$ close to the origin.
Integrating along the circle yields $g_{N-1}(u)$ and one arrives at the identity
\begin{equation}\label{6.22f}
g_N(u) = \int_{\Gamma_a}  \mathrm{d}z_2\cdots\int_{\Gamma_a} \mathrm{d}z_N  \int_{\tilde{\Gamma}_a}  \mathrm{d}z_1
\prod_{1 \leq i < j \leq N}\frac{z_j -  z_i}{z_j - \tau z_i} 
\prod_{j=1}^N  \frac{1}{z_j}    \mathrm{e}^{z_j(u - y_j)}\mathrm{e}^{\frac{1}{2}z_j^2t}
+g_{N-1}(u)\,.
 \end{equation} 
In the limit $u \to \infty$ the first summand vanishes, since all poles of the $z_1$-integration are to the right of 
$\tilde{\Gamma}_a$. Hence $\lim_{u \to \infty}g_N(u) = \lim_{u \to \infty}g_{N-1}(u)$.
But $\lim_{u \to \infty}g_1(u) = 1$ and the claim follows by induction. 
\end{proof}\\\\ 
This concludes the proof of Theorem \ref{th7}.

There are two limiting cases of interest, $\tau \to 1$ which corresponds to the symmetric interaction and $\tau \to 0$  
which corresponds to the maximally asymmetric interaction. In the limit $\tau \to 1$ one has $S(z_\alpha,z_\beta) = -1$.
\begin{cor}\label{cor11}
For $ \tau = 1$
\begin{equation}\label{6.23}
P_y(x,t;\tau = 1) = \mathrm{perm}\big( p_t(x_i - y_j)\big |_{i,j=1}^N\big)
\end{equation}
with the Gaussian kernel $p_t(u) = (2\pi t)^{-1/2}\exp(-u^2/2t)$ and $\mathrm{perm}$ denoting the permanent, \textit{i.e.}
omitting the factor $\mathrm{sgn}\,\sigma$ in the definition of the determinant.
\end{cor}
The contribution of Harris \cite{Har65} relies on the formula (\ref{6.23}). 
The limit $\tau \to 0$ of the transition probability has been first written down in \cite{SW99}, see also \cite{War07}.
\begin{cor}\label{cor12}
For $ q = 1$
\begin{equation}\label{6.24}
P_y(x,t;q = 1) = \mathrm{det}\big( F_{i-j}(x_i - y_j)\big |_{i,j=1}^N\big)\,,
\end{equation} 
where for $m\in \mathbb{Z}$
\begin{equation}\label{6.25}
F_m(u) =  \int_{\Gamma_a} \mathrm{d}z z^m\mathrm{e}^{
zu}\mathrm{e}^{\frac{1}{2}z^2t}\,.
\end{equation}
\end{cor}
\begin{proof} For $q =1$ the integrand in (\ref{6.6f}) reads
\begin{eqnarray}\label{6.27}
&&\hspace{-20pt}\sum_{\sigma \in S_N}  \mathrm{sgn}\,\sigma \prod_{1 \leq i < j \leq N}\frac{z_{\sigma(j)} }{z_j} \prod_{j=1}^N    \mathrm{e}^{
z_{\sigma(j)}x_j - z_jy_j}\mathrm{e}^{\frac{1}{2}z_j^2t}\nonumber\\
&&\hspace{40pt}=\sum_{\sigma \in S_N}  \mathrm{sgn}\,\sigma \prod_{1 \leq i < j \leq N}\frac{z_j }{z_{\sigma(j)}} \prod_{j=1}^N    \mathrm{e}^{
z_j(x_{\sigma(j)} - y_j)}\mathrm{e}^{\frac{1}{2}z_j^2t}\,.
\end{eqnarray}
Using  the identity
\begin{equation}\label{6.28}
\prod_{1 \leq i < j \leq N}\frac{z_j }{z_{\sigma(j)}} = \prod_{j=1}^N(z_j)^{\sigma(j) - j}
\end{equation}
results in (\ref{6.24}). 
\end{proof}

\begin{appendix}
\section{Appendix: Non-universal constants}\label{app.c}
\renewcommand{\thetheorem}{A.\arabic{theorem}}
\setcounter{equation}{0}
The asymptotics in (\ref{1.12}) is the sum of two terms. The deterministic term is proportional to $t$. Its prefactor
can be guessed on the basis of the Hamilton-Jacobi  equation for the height, 
\begin{equation}\label{b.1}
\partial_t  h = \gamma P(\partial_x h)\,,
\end{equation} 
$\gamma  = q-p$, compare with (\ref{1.10}).  
The solution to (\ref{b.1}) should be of the self-similar form, $h(x,t) =t\phi(x/t)$, for large $t$. Then the reference point is 
chosen as $x =ut$ 
and to leading order the height grows linearly in $t$. Such structure can be achieved for wedge initial conditions including the degenerate linear profile, $h(x,0) = \ell x$, which is referred to as either flat or stationary initial condition.
The fluctuating part of (\ref{1.12}) is more difficult. Here
our conjecture  relies on a particular model with exact solutions. The respective formula can be put in a form which makes its generalization evident and can be checked against a few other models. In fact, the conjectures are really based on the universality hypothesis for models in the KPZ class. In our context the hypothesis states that, for $\gamma \neq 0$, the fluctuation properties are independent of the choice of the interaction potential $V$, except for potential dependent scales. The non-universal prefactors listed below could possibly vanish, in which case a more detailed analysis is required.

We discuss separately the three canonical cases, wedge, flat, and stationary initial conditions. \medskip\\
\textit{(i) wedge initial conditions}. We consider two
initial wedges, labelled by $\sigma = +,-$ and given by
\begin{eqnarray}\label{b.2}
&&\hspace{0pt}h_+(x,0) = \ell_-x \,\,\mathrm{for}\,\, x \leq 0\,, \quad h_+(x,0) = \ell_+x \,\,\mathrm{for}\,\, x \geq 0\,,\\[1 ex]
\label{b.6}
&&\hspace{0pt}h_-(x,0) = \ell_+x\,\,\mathrm{for}\,\, x \leq 0\,, \quad h_-(x,0) = \ell_-x \,\,\mathrm{for}\,\, x \geq 0
\end{eqnarray} 
with $\ell_- < \ell_+$ and denote by $h_\sigma(x,t)$ the corresponding solution of (\ref{b.1}). Our initial value problem is equivalent to the Riemann problem for a scalar conservation law in one dimension, which is a well studied, see
 \cite{Ho}, Chapter 2.2, for a detailed discussion. 
 
We define
\begin{equation}\label{b.3}
\phi_+(y) = \sup_{\ell_- \leq \ell \leq \ell_+}\big(\ell y + \gamma P(\ell)\big)
\end{equation} 
and correspondingly 
\begin{equation}\label{b.4}
\phi_-(y) = \inf_{\ell_- \leq \ell \leq \ell_+}\big(\ell y + \gamma P(\ell)\big)\,.
\end{equation} 
$\phi_+$ is convex up and $\phi_-$ is convex down. $\phi_\sigma$ is linear outside the interval $[y^-_\sigma, y^+_\sigma] $ 
with slope $\ell_{-\sigma}$ to the left and $\ell_\sigma$ to the right of the interval. Inside the interval there are finitely many cusp points, \textit{i.e} shocks for the slope. We label them as $y^-_\sigma < y_\sigma^1 < ...
< y_\sigma^{k\sigma} < y^+_\sigma$, where the cases $y^-_\sigma < y^+_\sigma$, no cusp point, and $y^-_\sigma = y^+_\sigma$
are admitted. Then $h_\sigma(x,t)$ is self-similar and reads
\begin{equation}\label{b.5}
h_\sigma(x,t) = t \phi_\sigma(x/t)\,.  
\end{equation} 

We consider now the coupled diffusions $x_j(t)$, $j \in \mathbb{Z}$, governed by Eq. (\ref{1.8}). As initial measure we choose $x_0 = 0$, $x_{j+1} - x_j $, $j \geq 0$, independently distributed  according to (\ref{1.5}) with pressure 
$P(\ell_+)$, and $x_{j} - x_{j-1} $, $j \leq 0$, independently distributed  according to (\ref{1.5}) with pressure 
$P(\ell_-)$. For  case (\ref{b.6}) we impose the obviously interchanged initial conditions.
\begin{con}\label{con6}
Let $u \in \,]y^-_\sigma, y^+_\sigma[ $ and different from a cusp point. Furthermore set $\ell_0 = \phi_\sigma'(u)$, $A = - P'(\ell_0) > 0$, $\lambda = \gamma P''(\ell_0) \neq 0$. 
Then
\begin{equation}\label{b.7}
\lim_{t \to \infty}  \mathbb{P}\big(x_{\lfloor ut \rfloor}(t) - t \phi_\sigma (u)\leq 
- \mathrm{sgn}(\phi''_\sigma(u))(\tfrac{1}{2}|\lambda| A^2t)^{1/3} s\big) = F_{\mathrm{GUE}}(s)   \,.  
\end{equation}
\end{con} 
$\xi_{\mathrm{GUE}}$ has a negative mean and the actual interface is more likely located towards the interior of  tangent circle
at $(u, \phi_\sigma (u))$.
If $P''$ has a definite sign, then one of the two cases is empty. But in general either case has to be considered.

Our conjecture is based on the KPZ equation, from which the non-universal coefficients follow immediately by its scale invariance \cite{ACQ10,SS10}. The result has been confirmed by  the TASEP with step initial conditions \cite{J00} and a variety of similar
models \cite{BCF13,Sp12}.
\medskip\\
\textit{(ii) flat initial conditions}. If $\ell_- = \ell = \ell_+$, then the solution to (\ref{b.1}) reads $h(x,t) = \ell x + P(\ell)t$.
A natural microscopic choice would be the deterministic data $x_j(0) = \ell j$, as discussed in the Introduction. 
Such a microscopic configuration is called flat, since there are no deviating fluctuations from strict periodicity.
\begin{con}\label{con7}
For flat initial conditions with slope $\ell$, $\lambda = \gamma P''(\ell) \neq 0$, and $A = -P'(\ell)$,
\begin{equation}\label{b.8}
\lim_{t \to \infty}  \mathbb{P}\big(x_{\lfloor ut \rfloor}(t) - (u\ell + \gamma P(\ell))t\leq -\mathrm{sgn}(\lambda)
 (|\lambda| A^2t)^{1/3} s\big) = F_{\mathrm{GOE}}(2s)    
\end{equation} 
with   $F_{\mathrm{GOE}}(s) = \mathbb{P}(\xi_{\mathrm{GOE}} \leq s)$.
\end{con}
Note that, as in Conjecture \ref{con6}, the term linear in $t$ is dictated by the solution to the macroscopic equation. 
The non-universal scale coincides with one for the wedge. But the statistical properties of the fluctuations are distinct. They are now given by the Tracy-Widom GOE
 edge distribution, and more generally by the Airy$_1$ process, in contrast to the wedge, where one obtains GUE and the Airy$_2$ process. 
 
 Since there is no exact solution for the KPZ equation available, this time we use as reference model the TASEP
 with a periodic particle configuration as initial condition \cite{S04,BFPS05}. The resulting formula has been checked for a few other models \cite{BFS06,FSW13}.\medskip\\
\textit{(iii) stationary initial conditions}. A second choice for a macroscopically flat height profile is to make 
the increments $\{r_j, j \in \mathbb{Z}\}$ time stationary, see (\ref{1.7a}), (\ref{1.5}). 
\begin{con}\label{con9}
For stationary conditions with slope $\ell$
\begin{equation}\label{b.9}
\lim_{t \to \infty}  \mathbb{P}\big(x_{\lfloor -t \gamma P'(\ell)\rfloor}(t) - (-\ell P'(\ell) + P(\ell))\gamma t\leq - \mathrm{sgn}(\lambda)
 (\tfrac{1}{2}|\lambda| A^2t)^{1/3} s\big) = F_{\mathrm{BR}}(s) \,.   
\end{equation} 
\end{con}
The Baik-Rains distribution function,   $F_{\mathrm{BR}}(s)$, also denoted by $F_0(s)$, is defined in  \cite{BR00,Phome}. As far as known, it is not related to any of the standard matrix ensembles.  In (\ref{b.7}) and (\ref{b.8}) the reference point  $\lfloor ut \rfloor$ is arbitrary, while (\ref{b.9})
only close to the characteristic of Eq. (\ref{b.1})
one observes the anomalous $t^{1/3}$ scaling. Away from the characteristic the fluctuations would be Gaussian  generically. 

The asymptotics of the KPZ equation with stationary initial data has been accomplished recently  \cite{BCFV14}.
By scaling the result (\ref{b.9}) follows, which is then confirmed through the TASEP  \cite{BR00,PS00,FS05}
and the stationary version of the model defined in  (\ref{1.15}) \cite{BCFV14}.

\section{Appendix: Convergence to point-interaction}\label{app.a}
\renewcommand{\thetheorem}{B.\arabic{theorem}}
\setcounter{equation}{0}
We prove that point-interactions are approximated by a short range, sufficiently repulsive potential interaction.
To start we choose a potential $V \in C^2(\mathbb{R} \setminus \{0\}, \mathbb{R}_+)$  with the properties $V(u) = V(-u)$,
$\mathrm{supp} \,V = [-1,1]$, $V'(u) \leq 0$ for $u > 0$, and, for some $\delta > 0$, $\lim_{u \to 0} |u|^{\delta}V(u) > 0$. 
The scaled potential is defined by $V_\epsilon(u) = V(u/\epsilon)$. As in the introduction, we introduce the diffusion process, $x^\epsilon(t)$, governed by
\begin{eqnarray}\label{e.1}
&&\hspace{-10pt}\mathrm{d}x^\epsilon_0(t) = pV_\epsilon'( x^\epsilon_{1}(t)- x^\epsilon_{0}(t))\mathrm{d}t + \mathrm{d}B_0(t)\,,\nonumber\\
&&\hspace{-10pt}\mathrm{d}x^\epsilon_j(t) = \big(pV_\epsilon'( x^\epsilon_{j+1}(t)- x^\epsilon_{j}(t)) - qV_\epsilon'(x^\epsilon_{j}(t) - x^\epsilon_{j-1}(t))\big)\mathrm{d}t + \mathrm{d}B_j(t)\,,\quad j = 1,...,n-1\,,\nonumber\\
&&\hspace{-10pt}\mathrm{d}x^\epsilon_n(t) = - qV_\epsilon'(x^\epsilon_{n}(t) - x^\epsilon_{n-1}(t))\mathrm{d}t + \mathrm{d}B_n(t)\,.
\end{eqnarray}
The potential is entrance - no exit, hence $x^\epsilon(t) \in \mathbb{W}_{n+1}^+$ almost surely.
The limit process, $y(t)$, is governed by (\ref{2.1}),
\begin{eqnarray}\label{e.2}
&&\hspace{-10pt}y_0(t) = y_0 + B_0(t) - p \Lambda^{(0,1)}(t)\,, \nonumber\\
&&\hspace{-10pt}y_j(t) = y_j + B_j(t) - p \Lambda^{(j,j+1)}(t) +q \Lambda^{(j-1,j)}(t)\,,\quad j = 1,...,n-1\,,\nonumber\\
&&\hspace{-10pt}y_n(t) = y_n + B_n(t)  +q \Lambda^{(n-1,n)}(t)\,.
\end{eqnarray}
The processes $x^\epsilon(t), y(t)$ are defined on the same probability space.
\begin{theorem}\label{th12}
Let $x^\epsilon(t), y(t)$ be defined as in (\ref{e.1}), (\ref{e.2}) with $x^\epsilon(0) = y(0) 
 \in \mathbb{W}_{n+1}^+$. Then 
\begin{equation}\label{e.3}
\lim_{\epsilon \to 0}\mathbb{E}\big((x^\epsilon(t)- y(t))^2\big) = 0\,.
\end{equation}
\end{theorem}
\begin{proof}
We switch to relative coordinates, $r^\epsilon_0 = x^\epsilon_0$, $r^\epsilon_j =x^\epsilon_j  - x^\epsilon_{j-1}$, $u_0 = y_0$, $u_j = y_j - y_{j-1}$, $j = 1,...,n$. Then 
\begin{eqnarray}\label{e.4}
&&\hspace{-10pt}r^\epsilon_0(t) =  u_0(0)+ B_0(t) -p \Psi^\epsilon_1(t)\,,\nonumber\\
&&\hspace{-10pt}r^\epsilon_1(t) =  u_1(0) + B_1(t) - B_0(t) -p \Psi^\epsilon_2(t) +\Psi^\epsilon_1(t)\,,\nonumber\\
&&\hspace{-10pt}r^\epsilon_j(t) = u_j(0) + B_{j}(t) - B_{j-1}(t) -p\Psi^\epsilon_{j+1}(t) +\Psi^\epsilon_j (t)- q \Psi^\epsilon_{j-1}(t)\,,\quad j = 2,...,n-1\,,\nonumber\\
&&\hspace{-10pt}
r^\epsilon_n(t) = u_n(0) + B_{n}(t) - B_{n-1}(t)  +\Psi^\epsilon_n(t)- q \Psi^\epsilon_{n-1}(t)\,,
\end{eqnarray}
where
\begin{equation}\label{e.5}
\Psi^\epsilon_j(t) = - \int_0^t V_\epsilon'( r^\epsilon_{j}(s))  \mathrm{d} s\,.
\end{equation}
Correspondingly for the limit process,
\begin{eqnarray}\label{e.6}
&&\hspace{-10pt}u_0(t) =  u_0(0)+ B_0(t) -p \Lambda_1(t)\,,\nonumber\\
&&\hspace{-10pt}u_1(t) =  u_1(0) + B_1(t) - B_0(t) -p \Lambda_2(t) +\Lambda_1(t)\,,\nonumber\\
&&\hspace{-10pt}u_j(t) = u_j(0) + B_{j}(t) - B_{j-1}(t) -p\Lambda_{j+1}(t) +\Lambda_j (t)- q \Lambda_{j-1}(t)\,,\quad j = 2,...,n-1\,,\nonumber\\
&&\hspace{-10pt}
u_n(t) = u_n(0) + B_{n}(t) - B_{n-1}(t)  +\Lambda_n(t)- q \Lambda_{n-1}(t)\,,
\end{eqnarray}
with $\Lambda_j (t) = \Lambda^{(j-1,j)}(t)$ which depends only on $u_j(t)$.

On the right of (\ref{e.4}) and (\ref{e.6}) we note the T\"{o}plitz matrix $A$, $A_{ij} = -q \delta_{ij+1}  + \delta_{ij} -p\delta_{ij-1}$, $i,j = 1,...,n$. $A$ has the explicit inverse
\begin{equation}\label{e.7}
 (A^{-1})_{ij} = P_{ij}\,,\,\, \mathrm{for} \,\,1\leq i \leq j\,,\quad (A^{-1})_{ij} = P_{ji}\tau^{j-i}\,,\,\,\mathrm{for}\,\, j \leq i \leq n\,,
\end{equation}
with 
\begin{equation}\label{e.8}
P_{ij} =  \frac{(k_1^i - k_2^i )( k_1^{n+1 -j} - k_2^{n+1 -j})}{p(k_1 - k_2)(k_1^{n+1} - k_2^{n+1})}
\end{equation}
and $k_1,k_2$ the two real and distinct roots of $ -q + k -p k^2 = 0$. Thereby one confirms that there exists a $n\times n$ matrix $C$,  with $C = C^{\mathrm{t}}$, $C>0$, such that $CA = \mathrm{diag}(\tau^{0},...,
\tau^{n-1})$. 

Let us consider the quadratic form $\langle(\underline{r}^\epsilon(t)- \underline{u}(t)), C (\underline{r}^\epsilon(t)- \underline{u}(t))\rangle$, where $\underline{r}^\epsilon= (r^\epsilon_1,...,r^\epsilon_n)$,
$\underline{u}= (u_1,...,u_n)$.  The component $j = 0$ will be treated separately. Then
\begin{eqnarray}\label{e.9}
&&\hspace{0pt}\mathrm{d}\langle(\underline{r}^\epsilon(t)- \underline{u}(t)), C (\underline{r}^\epsilon(t)- \underline{u}(t))\rangle  = 2 \sum_{j=1}^{n}\tau^{j-1} (\underline{r}^\epsilon(t)- \underline{u}(t))_j(
 \mathrm{d}\Psi^\epsilon_j(t) - \mathrm{d}\Lambda_j(t))\nonumber\\
 &&\hspace{40pt}  \leq 2 \sum_{j=1}^{n}\tau^{j-1}\big(-r^\epsilon_j(t)  V_\epsilon'(r^\epsilon_j(t))\big) \mathrm{d}t\,,
\end{eqnarray}
since $r^\epsilon_j(t), u_j(t), \mathrm{d}\Psi^\epsilon_j (t), \mathrm{d}\Lambda_j (t)\geq 0$, and $u_j(t)\mathrm{d}\Lambda_j (t) = 0$. Using $\mathrm{supp}\, V_\epsilon = [-\epsilon,\epsilon]$, one arrives at 
\begin{equation}\label{e.10}
\langle(\underline{r}^\epsilon(t)- \underline{u}(t)), (\underline{r}^\epsilon(t)- \underline{u}(t))\rangle \leq - 2\epsilon \tau^{-n}\sum_{j=1}^{n}
\int _0^tV_\epsilon'(r^\epsilon_j(s))\mathrm{d}s\,.
\end{equation}

To deal with $r^\epsilon_0(t)$ one notes that
\begin{equation}\label{e.11}
r^\epsilon_0(t)  - u_0(t) = -p( \Psi^\epsilon_1(t) - \Lambda_1(t))\,,\quad
 \Psi^\epsilon_1(t) - \Lambda_1(t) =  [A^{-1}(r^\epsilon(t)- u(t))]_1\,.
\end{equation}
Thus the proof is completed, provided $\mathbb{E}(\Psi^\epsilon_j(t))$ is bounded uniformly in $\epsilon$.

For this purpose note that
\begin{equation}\label{e.12}
\mathbb{E}\big([A^{-1}(\underline{r}^\epsilon(t)- \underline{r}^\epsilon(0))]_j\big) = \mathbb{E}\big(\Psi^\epsilon_j(t)\big)\,.
\end{equation}
 We choose $f \in C^2(\mathbb{R_+})$, such that
 $f(r) = 1$ for $0 \leq r \leq 1$, $f(r) = r$ for large $r$ with smooth interpolation atisfying $f'(r) \ge 0$, $|f''(r)| \leq c_0$.
 Let $L_\epsilon$ denote the generator for $r^\epsilon(t)$. Then, for $j = 2,...,N-1$,
 \begin{equation}\label{e.13}
L_\epsilon f(r_j) = \big( qV'_\epsilon(r_{j+1}) - V'_\epsilon(r_{j}) + pV'_\epsilon(r_{j-1})\big)f' (r_j) + f''(r_j) 
\end{equation} 
and correspondingly for $j= 1,N$. Since $V'_\epsilon \leq 0$, $f' \geq 0$, and $V'_\epsilon f' = 0$, one arrives at the bound
 \begin{equation}\label{e.14}
\mathbb{E}\big(r_j^\epsilon(t)\big) \leq \mathbb{E}\big(f(r_j^\epsilon(t))\big) \leq f(r_j^\epsilon(0)) + c_0\,. 
\end{equation}
with some constant $c_0$ independent of $\epsilon$. Thus $\mathbb{E}(\Psi^\epsilon_j(t))$ is bounded  uniformly  in $\epsilon$.
 \end{proof}

\section{Appendix: Low density ASEP}\label{app.b}
\renewcommand{\thetheorem}{C.\arabic{theorem}}
\setcounter{equation}{0}

We explain an alternative proof of Theorem \ref{th1} based on ASEP duality.
\begin{pro}\label{pro13}
Let $\rho_+: \mathbb{W}_m^+ \to \mathbb{R}$, $\rho_-: \mathbb{W}_n^- \to \mathbb{R}$
be bounded, continuous probability densities. Then
\begin{equation}\label{a.8}
\int_ {\mathbb{W}_n^-}\int_{ \mathbb{W}_m^+} \mathrm{d} x\mathrm{d}y \rho_-(x)\rho_+(y)\mathbb{E}_y\big(H(x(t),y)\big)
 = \int_ {\mathbb{W}_n^-}\int_{ \mathbb{W}_m^+} \mathrm{d} x\mathrm{d}y \rho_-(x)\rho_+(y)\mathbb{E}_y\big(H(x,y(t))\big)  \,.                
\end{equation}
\end{pro}
\textit{Remark}.  Since by Theorem \ref{th7} the transition probability has a continuous density, one can take the limits 
$\rho_-(x) \to \delta(x - x_0)$, $\rho_+(y) \to \delta(y - y_0)$ and duality
holds in fact pointwise.\\\\ 
\begin{proof} We set
\begin{equation}\label{a.9}
H_+(y) = \int_ {\mathbb{W}_n^-}\mathrm{d} x \rho_-(x) H(x,y)\,,\quad
H_-(x) = \int_ {\mathbb{W}_m^+}\mathrm{d} y \rho_+(y) H(x,y)\,.
\end{equation} 
$H_+,H_-$ are continuous and (\ref{a.9}) reads 
\begin{equation}\label{a.10}
\mathbb{E}_{\rho_-}\big(H_-(x(t))\big) = \mathbb{E}_{\rho_+}\big(H_+(y(t))\big)\,.\medskip
\end{equation} 
\textit{(i) The approximation theorem}. It suffices to discuss the particle process $y(t)$. We consider $m$ ASEP particles with positions $w_1(t) <... <w_m(t)$, $w_j(t)  \in \mathbb{Z}$. Particles jump with rate $p$ to the right and rate $q$ to the left, subject to the exclusion rule. Switching to the moving frame of reference and under diffusive rescaling one obtains
\begin{equation}\label{a.11}
y_j^\epsilon(t) = \epsilon\big(w_j(\epsilon^{-2}t) - \lfloor(p-q)\epsilon^{-2}t\rfloor\big)
\end{equation}
with $\lfloor \cdot \rfloor$denoting integer part. Clearly $y_j^\epsilon(t)  \in (\mathbb{W}_m^+)^{\circ} \cap (\epsilon \mathbb{Z})^m$.
\begin{pro}\label{pro14}
 Let $f : \mathbb{W}_m^+ \to \mathbb{R}$ be bounded and continuous. Then for initial conditions $y^\epsilon$ such that $y^\epsilon \to y \in \mathbb{W}_m^+$ it holds
\begin{equation}\label{a.12}
\lim_{\epsilon \to 0} \mathbb{E}_{y^\epsilon}\big(f(y^\epsilon(t))\big)= \mathbb{E}_{y}\big(f(y(t))\big)\,.\medskip
\end{equation} 
\end{pro}
In \cite{KPS12}, the proposition is proved for the asymmetric zero range process with 
constant rate, $c(n) = 1 - \delta_{0n}$, which differs from the ASEP at most by $m$ uniformly in $t$.\\\\
\textit{(ii) ASEP duality}. We introduce $n$ dual particles. They jump with rate $q$ to the right and rate $p$ to the left,
subject to the exclusion rule. The diffusively rescaled positions of the dual particles in the moving frame
are denoted by $x_j^\epsilon(t)$.
\begin{pro}\label{pro15}
 For all $x \in (\mathbb{W}_m^+)^{\circ} \cap (\epsilon \mathbb{Z})^m$ and
$y \in (\mathbb{W}_m^+)^{\circ} \cap (\epsilon \mathbb{Z})^m$
 it holds 
\begin{equation}\label{a.13}
\mathbb{E}_x^{\epsilon}\big(H(x^{\epsilon}(t),y)\big) = \mathbb{E}_y^{\epsilon}\big(H(x,y^{\epsilon}(t))\big)\,.
\end{equation} 
\end{pro}
In \cite{BCS14} the assertion is proved for $\epsilon = 1$ at fixed lattice frame. In (\ref{a.13})
the $y^\epsilon(t)$ frame moves with velocity $p-q$, while the $x^\epsilon(t)$ frame with velocity $q-p$.
To check that the terms just balance one uses that $\theta(\lambda u) = \theta(u)$ for $\lambda > 0$
and the translation invariance of the ASEP dynamics.\bigskip\\
\textbf {Proof of Proposition \ref{pro13}}: In  (\ref{a.13}) we regard both sides as a piecewise constant function on $\mathbb{W}_m^+ \times\mathbb{W}_n^-$.
Integrating over $\rho_+ \times \rho_-$ yields
\begin{eqnarray}\label{a.14}
&&\hspace{-30pt}\int_ {\mathbb{W}_n^-}\int_{ \mathbb{W}_m^+} \mathrm{d} x\mathrm{d}y \rho_-(x)\rho_+(y)
\mathbb{E}_{\lfloor x \rfloor_{\epsilon}}\big(H(x^\epsilon(t),\lfloor y \rfloor_{\epsilon})\big)\nonumber\\
&&\hspace{20pt} = \int_ {\mathbb{W}_n^-}\int_{ \mathbb{W}_m^+} \mathrm{d} x\mathrm{d}y \rho_-(x)\rho_+(y)
 \mathbb{E}_{\lfloor y \rfloor_{\epsilon}}\big(H(\lfloor x \rfloor_{\epsilon},y^{\epsilon}(t))\big) 
\end{eqnarray} 
with $\lfloor \cdot \rfloor_{\epsilon}$ the integer part mod $\epsilon$. By continuity of $\rho_-$, $\rho_+$,
\begin{eqnarray}\label{a.15}
&&\hspace{-30pt}\int_ {\mathbb{W}_n^-}\int_{ \mathbb{W}_m^+} \mathrm{d} x\mathrm{d}y \rho_-(x)\rho_+(y)
\mathbb{E}_{\lfloor x \rfloor_{\epsilon}}\big(H(x^\epsilon(t),y )\big)\nonumber\\
 &&\hspace{20pt}= \int_ {\mathbb{W}_n^-}\int_{ \mathbb{W}_m^+} \mathrm{d} x\mathrm{d}y \rho_-(x)\rho_+(y)
 \mathbb{E}_{\lfloor y \rfloor_{\epsilon}}\big(H(x ,y^{\epsilon}(t))\big) +  o(\epsilon).
\end{eqnarray} 
Using Proposition \ref{pro14} establishes the claim. 
\end{proof}
\section{Appendix: Proof  of (\ref{6.6d})}\label{app.d}
\renewcommand{\thetheorem}{D.\arabic{theorem}}
\setcounter{equation}{0}
We fix $\sigma$, $\sigma \neq \mathrm{id}$, $n$, hence the sets $A$, $B$, and $\ell = \mathrm{min}\,B$.
We have argued already that the integration over $z_\ell$ results in an expression  vanishing as $t \to 0$.
To have a proof we have to study the full $2N$-dimensional integral.
For $f \in \mathcal{D}_\epsilon$, this integral reads
\begin{eqnarray}\label{d.2}
&&\hspace{-30pt} I_\sigma(y;f,t) = \int_{\mathbb{R}^N} \mathrm{d} x f(x) \int_{\Gamma_{aN}} \mathrm{d}z_0
  \int_{\Gamma_{a}} \mathrm{d}z_1 ... \mathrm{d}z_{N-1} \nonumber\\
&&\hspace{30pt} 
\times\prod_{j\in B} S(z_N,z_j)
\prod_{\{\alpha,\beta\} \in \mathrm{In}(\sigma), \,\alpha \neq N} S(z_\alpha,z_\beta) 
 \prod_{j=1}^N \mathrm{e}^{z_j(x_{\sigma^{-1}(j)} - y_j)} \mathrm{e}^{\frac{1}{2}z_j^2t} 
\end{eqnarray} 
with $z_0 = z_1 + ...+z_N$. The phase factor for $z_\ell$ is given by 
\begin{equation}\label{d.3}
\mathrm{e}^{z_\ell(x_{\sigma^{-1}(\ell)} - x_n + y_N - y_\ell)}\,.
\end{equation} 
By construction, $x_{\sigma^{-1}(\ell)} - x_n \geq \epsilon$ on the support of $f$. 
We introduce the change of variables
\begin{equation}\label{d.4}
w = x_{\sigma^{-1}(\ell)} - x_n\,,\quad w_j = x_{\sigma(j)} - x_n\,, \,j = 1,...,N\,,\,j \neq \ell\,, \quad w_0 = x_n \,.
\end{equation} 
Also, as shorthand, we  introduce 
$\underline{z}^{{\vee}\ell}
= (z_0, ..., \cancel{z_\ell}, ...,z_{N-1})$, $z^{{\vee}\ell} = z_0 -z_1 - ...- \cancel{z_\ell} - ...-z_{N-1}$,
$\underline{z}^{{\vee}\ell,j} = (z_0, ..., \cancel{z_\ell},  \cancel{z_j}, ...,z_{N-1})$, $z^{{\vee}\ell,j} = z_0 -z_1 - ...- \cancel{z_\ell} -\cancel{z_j}- ...-z_{N-1}$.  Then
\begin{eqnarray}\label{d.5}
&&\hspace{-30pt} I_\sigma(y;f,t) = \int_{\mathbb{R}^{N-1}} \mathrm{d} \underline{w}^{{\vee}\ell}  \int \mathrm{d} 
\underline{z}^{{\vee}\ell}  \int_{\Gamma_{a}}\mathrm{d}z_\ell
  \int_\epsilon^\infty \mathrm{d}w \tilde{f}(w, \underline{w}^{{\vee}\ell})  \mathrm{e}^{z_\ell(w +y_N- y_\ell)}  \mathrm{e}^{\frac{1}{2}z_\ell^2t}\nonumber\\
&&\hspace{30pt} 
\times\prod_{j\in B} S(z_N,z_j)
\prod_{\{\alpha,\beta\} \in \mathrm{In}(\sigma), \,\alpha \neq N} S(z_\alpha,z_\beta) 
 \prod_{j=1,j \neq \ell}^N \mathrm{e}^{z_j(x_{\sigma^{-1}(j)} - y_j)} \mathrm{e}^{\frac{1}{2}z_j^2t}\,,
\end{eqnarray} 
where $\tilde{f}$ denotes $f$ under the linear transformation (\ref{d.4}).

The strategy is to first integrate over $\underline{w}^{{\vee}\ell}$ which results in $g(w,\underline{z}^{{\vee}\ell})$, where by construction $g$ is supported in $[\epsilon, \infty)$ in dependence on $w$ and is smooth with a rapid decay on the contours $\Gamma_a,\Gamma_{aN}$. Secondly we bound the integration in $ \mathrm{d}z_\ell
 \mathrm{d}w$ with an explicit dependence on $\underline{z}^{{\vee}\ell}$. For this purpose we have to study the $S$-factors.
 One has
 \begin{equation}\label{d.6}
\frac{\tau z_N - z_\ell}{\tau z_\ell - z_N} = \frac{\tau z^{{\vee}\ell} - (1 +\tau) z_\ell}{(1+\tau)z_\ell - z^{{\vee}\ell}} 
\end{equation} 
and for $j \in B\setminus \{\ell\}$
\begin{equation}\label{d.7}
\frac{\tau z_N - z_j}{\tau z_j - z_N} = \frac{\tau z_\ell - (1+\tau)z_j - \tau z^{{\vee}\ell,j}}{z_\ell + (1+\tau) z_j - z^{{\vee}\ell,j}}\,.
\end{equation} 
The integrand for $z_\ell$ has the form
\begin{equation}\label{d.8}
\prod _{j\in B \cup A(\ell)} \frac{z_\ell +a_j}{z_\ell - b_j}\,,
\end{equation} 
with $A(\ell) \subset A$, $a_j,b_j$ linear in $z^{{\vee}\ell}$, and $\Re(b_j) > a$.
For the remaining factors one only uses the bound
\begin{equation}\label{d.9}
|S(z_\alpha,z_\beta)| \leq c(1 + |z_\alpha| + |z_\beta|)
\end{equation} 
on $\Gamma_a, \Gamma_{aN}$.
\begin{lemma}\label{le15}
Let $a_j, b_j \in \mathbb{C}$ and $\Re(b_j) > a$, $j=1,...,m$, and define
\begin{equation}\label{d.10}
I(t) = \int_\epsilon^\infty \mathrm{d}w f(w) \int_{\Gamma_a}\mathrm{d}z\,
\mathrm{e}^{\frac{1}{2}z^2t}\mathrm{e}^{zw}
 \prod _{j = 1}^m \frac{z +a_j}{z - b_j}
\end{equation} 
for $f\in \mathcal{D}_\epsilon$. Then, uniformly in $t$, $0 \leq t \leq 1$,
\begin{equation}\label{d.11}
|I(t)| \leq c \prod_{j =1}^m (1+ |a_j| + |b_j|)\, \quad and \quad \lim_{t \to 0} I(t) = 0\,.
\end{equation}
\end{lemma} 
\begin{proof} 
The $z$-integrand is a product of $\mathrm{e}^{zw}$ and
\begin{equation}\label{d.12}
F_0(z) = \mathrm{e}^{\frac{1}{2}z^2t}\,, \quad  F_j(z) = \frac{z +a_j}{z - b_j} \,.
\end{equation} 
As distributions we define
\begin{eqnarray}\label{d.13}
&&\hspace{-8pt}\hat{F}_0(w) = \int_{\Gamma_a}\mathrm{d}z\,\mathrm{e}^{\frac{1}{2}z^2t}\mathrm{e}^{zw} = p_t(w)\,,\\ \label{d.14}
&&\hspace{-20pt}\quad \hat{F}_j(w) = \int_{\Gamma_a}\mathrm{d}z\,
\mathrm{e}^{zw}\frac{z +a_j}{z - b_j} = -\theta(-w)\mathrm{e}^{b_jw} + (a_j + b_j) \delta(w)\,.
\end{eqnarray} 
Then $I(t)$ is expressed as an $(m+1)$-fold convolution,
\begin{equation}\label{d.15}
I(t) = \int_\epsilon^\infty \mathrm{d}w f(w) (F_0*F_1*\cdots*F_m)(w)\,. 
\end{equation} 
Since $|\theta(-w)\mathrm{e}^{b_jw}| <1$, one obtains the bound of (\ref{d.11}). $\hat{F}_j$ is supported on
$(-\infty, 0]$, $f$ on $[\epsilon,\infty)$, and $\lim_{t \to 0} p_t(w) = \delta(w)$, which establishes the limit of
(\ref{d.11}).\medskip
\end{proof}

Next note that 
\begin{equation}\label{d.16}
|S(z_\alpha,z_\beta)| \leq c(1+ |z_1|+|z_2|)\,.
\end{equation} 
Hence
\begin{eqnarray}\label{d.17}
&&\hspace{-20pt} \big|  \int_{\Gamma_{a}}\mathrm{d}z_\ell
  \int_\epsilon^\infty \mathrm{d}w g(w,\underline{z}^{{\vee}\ell})  \mathrm{e}^{z_\ell(w +y_N- y_\ell)}  \mathrm{e}^{\frac{1}{2}z_\ell^2t} \big| \big|\prod_{j\in B} S(z_N,z_j)\prod_{\{\alpha,\beta\} \in \mathrm{In}(\sigma), \,\alpha \neq N} S(z_\alpha,z_\beta) 
\nonumber\\
&&\hspace{20pt} 
\times
 \prod_{j=1,j \neq \ell}^N \mathrm{e}^{z_j(x_{\sigma^{-1}(j)} - y_j)} \mathrm{e}^{\frac{1}{2}z_j^2t}\big| \leq P_N 
  (|\underline{z}^{{\vee}\ell}|)\sup_{w}|g(w,\underline{z}^{{\vee}\ell}) |
\end{eqnarray}
uniformly in $t$ with some polynomial $P_N$ at most of order $N$. Thus we can use dominated convergence to conclude that
\begin{equation}\label{d.18}
\lim_{t \to 0}  I_\sigma(y;f,t)=0 \,.
\end{equation}

\end{appendix}

\end{document}